\documentclass[10pt,a4paper]{article}
\usepackage[utf8]{inputenc}
\usepackage[english]{babel}
\usepackage{amsmath}
\usepackage{amsfonts}
\usepackage{amssymb}
\usepackage{dsfont}
\usepackage{amsthm}
\usepackage{float}
\usepackage{textcomp}
\usepackage{faktor}
\usepackage{xcolor}
\usepackage{graphicx}
\usepackage{amssymb}
\usepackage{fancybox}
\usepackage{mathrsfs}
\usepackage{mathtools}
\usepackage{tikz}
\usepackage{centernot}
\usepackage{esint}
\usepackage{setspace}
\usepackage{pgfplots}
\usepackage{multicol}
\usepackage{authblk}
\usepackage{appendix}
\usepackage{fullpage}
\usepackage{enumerate}
\usepackage{algorithm}
\usepackage[noend]{algpseudocode}
\usepackage{hyperref}
\usepackage{comment,verbatim}
\usetikzlibrary{arrows}
\DeclareUnicodeCharacter{0102}{\u{a}}

\theoremstyle{plain}
\newtheorem{thm}{Theorem}[section]
\newtheorem{lem}[thm]{Lemma}
\newtheorem{prop}[thm]{Proposition}
\newcommand{\thistheoremname}{}
\newtheorem*{genericthm*}{\thistheoremname}
\newenvironment{namedthm*}[1]
  {\renewcommand{\thistheoremname}{#1}%
   \begin{genericthm*}}
  {\end{genericthm*}}
\newtheorem{cor}[thm]{Corollary}

\theoremstyle{definition}
\newtheorem{defn}[thm]{Definition}

\newtheorem{prob}{Problem}

\newtheorem{assumption}{Assumption}

\theoremstyle{remark}
\newtheorem{rem}[thm]{Remark}

\DeclareMathOperator{\supp} {supp}
\DeclareMathOperator{\w} {w}
\DeclareMathOperator{\dd} {d}

\DeclareMathOperator{\Span} {span}

\DeclareMathOperator{\ev}{ev}
\DeclareMathOperator{\Res}{Res}
\DeclareMathOperator{\Ker}{Ker}

\newcommand{\AGcode}[3]{\mathcal{C}_L (#1, #2, #3)}

\newcommand{\word}[1]{\ensuremath{\boldsymbol{#1}}}
\newcommand{\xv}{\word{x}}
\newcommand{\yv}{\word{y}}
\newcommand{\error}{\word{e}}
\newcommand{\av}{\word{a}}
\newcommand{\bv}{\word{b}}
\newcommand{\cv}{\word{c}}
\renewcommand{\error}{\word{e}}

\newcommand{\wv}{\word{w}}
\newcommand{\zerov}{\word{0}}

\newcommand{\X}{\mathcal{X}}
\newcommand{\eqdef}{\stackrel{\textrm{def}}{=}}
\newcommand{\map}[4]{
    \left\{
    \begin{array}{ccc}
        #1 & \longrightarrow & #2 \\ #3 & \longmapsto & #4
    \end{array}
    \right.
}

\newcommand{\ip}[1]{\textcolor{red}{\bf [Isabella : #1 ]}}

\title{Attaining Sudan's decoding radius with no genus penalty for algebraic geometry codes}
\author[1,2]{Isabella Panaccione\thanks{\texttt{isabella.panaccione@inria.fr}}}
\affil[1]{INRIA}
\affil[2]{
  LIX, CNRS UMR 7161\break
  École Polytechnique,\break
  91128 Palaiseau Cedex, France
}

\begin{document}

\maketitle

\begin{abstract}
\noindent
In this paper we present a decoding algorithm for algebraic geometry codes with error--correcting capacity beyond half the designed distance of the code. This algorithm comes as a fusion of the Power Error Locating Pairs algorithm for algebraic geometry codes and the technique used by Ehrhard in order to correct these codes up to half the designed distance. The decoding radius of this algorithm reaches that of Sudan algorithm, without any penalty given by the genus of the curve.
\end{abstract}

\medskip

\noindent {\bf Key words : } Error correcting codes; algebraic geometry codes; decoding algorithms; error correcting pairs; Sudan algorithm; genus.

\section*{Introduction}
Algebraic geometry codes were first introduced by Goppa in $1981$ \cite{G81} and gave a breakthrough in coding theory when Tsafsman, Vl\u{a}dut and Zink proved that Gilbert Varshamov bound could be exceeded when some specific curves where considered \cite{TVZ82}. Furthermore, these codes have interested the Cryptography scene too, in particular for McEliece scheme \cite{JM96}.
\begin{comment} 
\noindent
About the decodability,  all cited algorithms for Reed--Solomon codes, can be generalised to the class of algebraic geometry codes at the price of a penalty in the genus of the curve for the decoding radius\end{comment}
\subsection*{Unique decoding of algebraic geometry codes}
Thanks to their strong algebraic structure, it has been possible to design several decoding algorithms for algebraic geometry codes. In $1989$ Justesen, Larsen, Jensen, Havemose and H{\o}holdt proposed one of the first decoding algorithms for a specific class of algebraic geometry codes \cite{JLJHH89} achieving the correction capacity of $\frac{d^*-1-g}{2}$, where $d^*$ is the designed distance of the code and $g$ the genus of the curve. This algorithm, also called \textit{basic algorithm} in the literature, is the starting point for several years of research aimed at improving this decoding radius by erasing the penalty in the genus of the curve. One of the first attempts in this sense, came from Skorobatov and Vl\u{a}dut \cite{SV90} who generalised the basic algorithm to arbitrary curves and improved the decoding radius in some cases. Their result was in turn improved by Duursma \cite{D93} who generalised it to all algebraic geometry codes and reached the decoding radius $\frac{d^*-1}{2}-\sigma$, where $\sigma$ is the \textit{Clifford defect}. The problem though was not complitely solved, as for instance for plane curves we have in average $\sigma=\frac{g}{4}$. This last algorithm is also referred to as the \textit{modified algorithm}. \\

\noindent
In parallel to the basic algorithm and with the same correction capacity, we have the algorithm proposed by Porter in his thesis \cite{P88}. Porter's idea mainly consists in solving a \textit{key equation}, using a generalisation of Euclide's algorithm for functions on curves. Though, the price of this generalisation lies in strong restrictions on the codes and the curves, which entail the correctness of the algorithm only for a small class of codes. In \cite{E92}, Ehrhard generalised this algorithm to all curves by solving the key equation of Porter's algorithm with simple linear algebra operations and proved this algorithm to be equivalent to the basic algorithm for a divisor $F$ with no evaluation points in its support. The correctness of the algorithm was proved independently as well by Porter, Shen and Pellikaan in  \cite{PSP92}, where in addition they succeded in pushing the decoding radius up to the one of the modified algorithm.\\

\noindent
The first algorithm able to correct up to $\frac{d^*-1}{2}$ has been proposed by Pellikaan in \cite{P89}. This algorithm, whose correctness is ensured for maximal curves and some other cases, consists in running the basic algorithm in parallel on several divisors $F_1, \dots, F_s$, as for counting reasons one among them has to work. Though, this result only ensures the existence of these divisors and even if it was extended to almost all curves (Vl\u{a}dut \cite{V90}), no practical precedure to find the $F_i$'s has been found yet. A constructive algorithm to achieve the correction capacity $\frac{d^*-1}{2}$ was finally found by Ehrhard in \cite{E93}, whose idea consists in providing a more suitable divisor $F$, obtained from a gradual adaptation process, and running \cite{E92} on this $F$. In the same year Feng and Rao proposed the so called \textit{majority vote} for \textit{unknown syndromes} \cite{FR93}, which also corrects $\frac{d^*-1}{2}$ errors.

\medskip
\noindent
In 1992 Pellikaan \cite{P92} and, independently, K\"oetter \cite{K92} introduced the so called \textit{error correcting pairs} algorithm, which generalises the basic algorithm to all codes which dispose from a particular structure called \textit{error correcting pair}. This algorithm cannot correct more than $\frac{d-1}{2}$ errors and in particular, since for algebraic geometry codes it reduces to the basic algorithm, for this class of codes it is equivalent to Ehrhard algorithm \cite{E92} for a divisor $F$ with no evaluation points in its support and corrects up to $\frac{d^*-1-g}{2}$ errors.
\subsection*{Beyond half the designed distance}
It is known that several decoding algorithms have been extended from Reed--Solomon codes to algebraic geometry codes to correct amounts of errors superior than half the designed distance. Though, as for the basic algorithm, anytime such a generalisation is made, a penalty in the genus of the curve appears in the decoding radius. Sudan algorithm \cite{S97} has been extended to algebraic geometry codes by Shokrollahi Wasserman \cite{SW99} with a penalty of $\frac{\ell g}{\ell+1}$, where $\ell$ is the degree of Sudan's polynomial. It is known that Sudan algorithm gives in return the list of all possible solutions to the decoding problem. Though it is possible to generalise to algebraic geometry codes also algorithms like the power decoding (\cite{SSB10}, revisited in \cite{R15}), which gives back the closer solution (if it exists) or fails. Also the version of the error correcting pairs algorithm to correct more errors, which is the \textit{power error locating pairs} algorithm \cite{CP20} can be run on algebraic geometry codes with a penalty in the genus $\frac{\ell g}{\ell+1}$, where the paramter $\ell$ is the \textit{power} used in the algorithm. These three algorithms have in practice the same decoding radius (at most they differ by $1$) and in particular they share the penalty in the genus. Now, it is known it is possible to get an improved decoding radius with no genus penalty using Guruswami--Sudan algorithm for an appropriate multiplicity. However, since the size of the linear system of the algorithm depends on this multiplicity, the complexity of the algorithm can become quite large very fast.

\subsection*{Our contribution}
In this paper we propose a decoding algorithm for algebraic geometry codes, whose decoding radius turns to equal the one of Sudan algorithm, but with no factor in $g$. To do so, first we adapt the language of power error locating pairs algorithm \cite{CP20} to the one of Ehrhard's result \cite{E93}. In this way we get a decoding algorithm with a correction capacity equivalent to Sudan's. Therefore, to erase the penalty, we apply our algorithm to a suitable divisor $F$ provided by using Ehrhard's adaptation process (\cite{E93}).

\subsection*{Outline of the article}
In \S\ref{Notation} we will give some notations and results about Riemann--Roch spaces, algebraic geometry codes and star product. In \S\ref{sec:Ehrhard}, also introductive, we will give an overview of Ehrhard's algorithm \cite{E93} and show it as an extension of \cite{E92} (see Remark \ref{rem:9293}). In \S\ref{sec:l2} and \S\ref{sec:l>2} the new algorithm is presented, while some experimental observations are given in \S\ref{sec:test}. The nature of these results being quite technical, we decided to report some proofs in appendix and suggest the reader to postpone their reading to a second moment.

\section{Notation and preliminaries}\label{Notation}
\subsection{Codes and decoding problems}
Given a finite field $\mathbb{F}_q$, a code over $\mathbb{F}_q$ of length $n$ is simply a subset of $\mathbb{F}_q^n$. The code is said to be \textit{linear} if it is a vector subspace of $\mathbb{F}_q^n$. Its elements are called \textit{codewords}. 
\begin{defn}
Given two vectors $\av=(a_1, \dots, a_n), \bv=(b_1, \dots, b_n)\in\mathbb{F}_q^n$, their \textit{Hamming distance} is 
\[\dd(\av, \bv)\eqdef \#\{i\in \{1, \dots, n\}\mid a_i\ne b_i\}.\] 
The weight of $\av$ is defined as $\w(\av)\eqdef\dd(\av, \zerov)$. 
\end{defn}
\begin{defn}
Given a code $\mathcal{C}\subseteq \mathbb{F}_q^n$, the \textit{minimum distance} of $\mathcal{C}$ is 
\[\dd(\mathcal{C})\eqdef\min_{\substack{\av, \bv\in\mathcal{C} \\ \av\ne\bv}} \dd(\av, \bv).\]
\end{defn}
\medskip
\noindent
In the rest of the paper we will write sometimes $\dd$ instead of $\dd(\mathcal{C})$ when there is no ambiguity on the code. We recall that if the code $\mathcal{C}$ is linear, which is the case for the codes considered in this paper, we have $\dd(\mathcal{C})=\min_{\av\in \mathcal{C}\setminus \{\zerov\}}\w(\av)$.
\begin{defn}
Given a vector $\av=(a_1, \dots, a_n)\in \mathbb{F}_q^n$ we define its \textit{support} as 
\[\supp(\av)\eqdef\{i\in \{1, \dots, n\}\mid a_i\ne 0\}.\]
\end{defn}
\medskip
\noindent
We can now present the decoding problem we want to solve. 
\begin{prob}\label{Dec_prob}
Let $\mathcal{C}\subseteq\mathbb{F}_q^n$ be a code, $\yv\in\mathbb{F}_q^n$ and $t\in\{1,\dots, n\}$. Return (if it exists) $\cv\in \mathcal{C}$ such that 
\[\dd(\yv, \cv)\le t.\] 
\end{prob}
\medskip
\noindent
This problem takes the name of \textit{bounded decoding problem}. A \textit{decoding algorithm} for a code $\mathcal{C}$, is an algorithm which solves Problem \ref{Dec_prob} for $\mathcal{C}$ for some $t$. The maximum value of $t$ for which the algorithm can solve this problem, is called \textit{decoding radius} of the algorithm. Depending on the value of $t$, it is possible to estimate the number of possible solutions to Problem \ref{Dec_prob}. In particular, it is known that if $t\le\frac{\dd(\mathcal{C})-1}{2}$, then there exists at most one solution. Once $t$ exceeds this amount, called \textit{unique decoding bound}, the uniqueness of the solution is no longer entailed. This last case is the one we want to treat for algebraic geometry codes.
\begin{rem}
Mind that it is not always easy to compute the exact minimum distance of an algebraic geometry code. Only lower bounds are known and the main one is referred to as the \textit{designed distance}, which is why for this codes the unique decoding bound is considered to be half the designed distance of the code. We will present the notion of algebraic geometry codes and designed distance in the next section.
\end{rem}
\medskip
\noindent
There exist several decoding algorithms which solve Problem \ref{Dec_prob} for algebraic geometry codes for $t$ up to half the designed distance and even beyond. We know that in the second case, there could be more than one solution, and it is not always easy to chose the ``best'' one. Indeed there could be cases, called \textit{worst cases}, where two solutions $\cv_1, \cv_2$ exist and satisfy $\dd(\cv_1, \yv)=\dd(\cv_2, \yv)$. Some decoding algorithms treat these cases by giving back the list of all possible solutions and then take the name of \textit{list decoding algorithm}. Other algorithms, like the one we are going to propose in this paper, are called \textit{probabilistic} and give back one solution or fail. For sake of simplicity, all along this paper we work on Problem \ref{Dec_prob} for a generic $t$, by making the following assumption.
\begin{assumption}\label{ass0}
There exists $\cv\in \mathcal{C}$ and $\error\in\mathbb{F}_q^n$ with $\w(\error)=t$, such that 
\[\yv=\cv+\error.\]
\end{assumption}
\medskip
\noindent
That is, we assume Problem \ref{Dec_prob} to have a solution $\cv$ and that $\dd(\cv, \yv)=t$.

\subsection{Algebraic geometry codes}
In what follows, we will only consider smooth projective geometrically connected curves over $\mathbb{F}_q$. For the sake of semplicity, no proof will be included in this first section, but we direct the reader to \cite{S09} and \cite{TVN07} for further details. The Riemann Roch space associated to a curve $\X$ and a divisor $G$ is defined as
\[L(G)\eqdef\{f\in \mathbb{F}_q(\X)^*\mid (f)\ge -G\}\cup\{0\}.\]
It is in particular a vector space over $\mathbb{F}_q$ and its dimension is denoted by $\ell(G)$. Depending on the degree of $G$ it is possible to deduce some informations about the dimension of the space $L(G)$.
\begin{prop}\label{prop:preRR}
The following properties hold
\begin{itemize}
\item If $\deg G<0$, then $L(G)=\{0\}$;
\item $\ell(G)\ge \deg G-g+1$;
\item if $\deg G>2g-2$, then $\ell(G)=\deg G-g+1$.
\end{itemize}
\end{prop}
\begin{thm}[Clifford's Theorem]\label{thm:clifford}
For all divisors $A$ with $0\le \deg A\le 2g-2$ holds 
\[\ell(A)\le 1+\frac{1}{2}\deg A.\]
\end{thm}
\medskip
\noindent
The reader can find the proof of this result in \cite[\S$1.6$]{S09}. In the following sections, we will work consistently with Riemann Roch spaces and their dimension. First, we recall that the \textit{support} of a divisor $A=\sum_P v_P(A)P$ is the set $\supp(A)\eqdef\{P\mid v_P(A)\ne 0\}$. Given two divisors $A, B$ it is possible to define the minimum (and symmetrically the maximum) of $A$ and $B$ 
\[\min\{A, B\}\eqdef \sum_P\min\{v_P(A), v_P(B)\} P.\]
 We recall the following properties:
\begin{itemize}
\item $L(A)\cap L(B)=L(\min\{A, B\})$,
\item $L(A) + L(B)\subseteq L(\max\{A, B\})$.
\end{itemize}
Finally, we will need the following results which bound the dimension of the Riemann Roch space of a sum of divisors. 
\begin{prop}\label{prop:ell}
Let $A, B$ be two divisors with $B\ge 0$. Then 
\begin{itemize}
\item[(\textit{i})] $\ell(A -B)\ge \ell(A)-\deg B$
\item[(\textit{ii})] $\ell(A-B)\le \max\{0, \ell(A)-\ell(B)+1\}$
\end{itemize}
\end{prop}
\medskip
\noindent
For the proof of (\textit{i}) see \cite[Lemma $1.4.8$]{S09}. For (\textit{ii}), one can observe that if $\ell(A-B)=0$ or $\ell(B)=0$, then the inequality is clearly true, while the proof for $\ell(A-B),  \ell(B)>0$ can be found in \cite[Lemma $1.6.14$]{S09}.

\medskip
\noindent
It is now possible to introduce algebraic geometry codes, whose notion relies indeed on the one of Riemann Roch space. Given a sequence of rational points $\mathcal{P}=(P_1, \dots, P_n)$ of a curve $\X$, the algebraic geometry code associated to $\X$, the divisor $G$ and $\mathcal{P}$ is defined as
\[\AGcode{\X}{\mathcal{P}}{G}=\{(f(P_1), \dots, f(P_n))\mid f\in L(G)\}.\] In the rest of the paper we will often use the notation $\ev_{\mathcal{P}}(f)$ to indicate the vector $(f(P_1), \dots, f(P_n))$. Given an algebraic geometry code, we will call the $P_i$'s, the \textit{evaluation points} of the code and we will denote by $D$ the divisor $P_1+\dots +P_n$. Furthermore, given the evaluation points $\mathcal{P}$ of an algebraic geometry code, we denote by $\omega$ a differential form such that $v_{P}(\omega)=-1$ and $\Res_P(\omega)=1$ for all $P\in \mathcal{P}$ and by $W$ the canonical divisor associated to $\omega$, that is $W=(\omega)$. It is known that, given such a $W$ for an algebraic geometry code $\AGcode{\X}{\mathcal{P}}{G}$
\begin{equation}\label{dualAG}
\AGcode{\X}{\mathcal{P}}{G}^{\perp}=C_{\Omega}(\X, \mathcal{P}, G)=\AGcode{\X}{\mathcal{P}}{W+D-G}.
\end{equation}
A proof of this result is given in \cite[Proposition $2.2.10$]{S09}. As it is known, it is possible to estimate the dimension and the minimum distance of an algebraic geometry code. We recall here the properties we will need in the following sections.
\begin{prop}
Let $\mathcal{C}=\AGcode{\X}{\mathcal{P}}{G}$, with $\deg G <n$. Then we have $\dim \mathcal{C}=\ell(G)$. In particular,
\begin{itemize}
\item $\dim \mathcal{C}\ge \deg G-g+1$ and $d(\mathcal{C})\ge n-\deg G$;
\item if $\deg G >2g-2$, then $\dim \mathcal{C}=\deg G-g+1$.
\end{itemize}
\end{prop}
\medskip
\noindent
The quantity $n-\deg G$ takes the name of \textit{designed distance} of the code $\mathcal{C}$ and we denote it by $d^*$.

\subsection{Star product}
Given two vectors $\av=(a_1, \dots, a_n)$, $\bv=(b_1, \dots, b_n)\in\mathbb{F}_q^n$, the \textit{star product} (also called Schur product) of $\av$ and $\bv$ is defined as 
\[\av\ast\bv\eqdef (a_1b_1, \dots, a_nb_n).\]
We denote by $\av^i$ the power with respect to the star product of the vector $\av$, that is $\av^i\eqdef (a_1^i, \dots, a_n^i)$. One should be careful not to mix the notions of star product $\av\ast\bv$ and canonical inner product in $\mathbb{F}_q^n$, $\langle \av, \bv\rangle=\sum_{i=1}^n a_ib_i$. It is easy to prove that the following relation between the two operations holds:
\begin{equation}\nonumber
\langle \av\ast\bv, \cv\rangle=\langle\av, \bv\ast\cv\rangle.
\end{equation}
It is possible to generalise the notion of star product to subsets of $\mathbb{F}_q^n$ as well. 
\begin{defn}
Given $A, B\subseteq \mathbb{F}_q^n$, the star product of $A$ and $B$ is defined as
\[A\ast B\eqdef \Span_{\mathbb{F}_q}\{\av\ast\bv\mid \av\in A, \bv\in B\}.\]
Finally, given $i\in \mathbb{N}$, the power $A^i$ is defined by induction as $A^i\eqdef A\ast A^{i-1}$, where $A^1\eqdef A$.
\end{defn}
\subsection{Star product and algebraic geometry codes}
Riemann Roch spaces, and so algebraic geometry codes, behave well with respect to Schur product. In particular given two divisors $F$ and $G$, it is easy to see that $L(F)L(G)\subseteq L(F+G)$, hence
\[
    \AGcode{\X}{\mathcal{P}}{G} * \AGcode{\X}{\mathcal{P}}{G'}
    \subseteq \AGcode{\X}{\mathcal{P}}{G+G'}.
  \]
Under some further condition we can have the equality.
\begin{prop}[Star product of AG codes]\label{prop:star_prod_AGcodes}
  Let $\X$ be a curve of genus $g$, $\mathcal{P} = (P_1, \dots, P_n)$
  be a sequence of rational points of $\X$ and $G, G'$ be two divisors of
  $\X$ such that $\deg G \geq 2g$ and $\deg G' \geq 2g+1$. Then,
  \[
    \AGcode{\X}{\mathcal{P}}{G} * \AGcode{\X}{\mathcal{P}}{G'}
    = \AGcode{\X}{\mathcal{P}}{G+G'}.
  \]
\end{prop}

\begin{proof}
  This is a consequence of \cite[Theorem 6]{M70}.
  For instance, see \cite[Corollary 9]{CMP17}.
\end{proof}

\section{Ehrhard's algorithm}\label{sec:Ehrhard}
In this section, we recall the version of Porter's algorithm for unique decoding proposed by Ehrhard \cite{E93}. In particular we report the adaptation process which premits to push the decoding radius up to half the designed distance of the code with no genus penalty, which is
\[\frac{d^*-1}{2}\cdot\]
Let us consider a code $\AGcode{\X}{\mathcal{P}}{G}\subseteq\mathbb{F}_q^n$, where $g-1\le\deg(G)\le n$ and $\supp(G)\cap\mathcal{P}=\emptyset$. We recall that by Assumption \ref{ass0} the received vector is of the form $\yv=\cv+\error$, where $\cv=\ev_D(f_{\cv})$ with $f_{\cv}\in L(G)$ and $t=\w(e)$. In particular, since this is an algorithm for unique decoding, we suppose $t\le \frac{d^*-1}{2}$. Note that in Ehrhard's paper the used language is the one of $C_{\Omega}$ codes, while we translated everything into $C_L$ codes.

\subsection{Foundations and purpose of the algorithm}\label{subs:found1}
First, we would like to express all vectors in $\mathbb{F}_q^n$ as vectors of evaluations of certain functions. In order to do that, we introduce a divisor $G'$ such that $\supp(G')\cap\mathcal{P}=\emptyset$, $G'\ge G$ and $\ell(W+D-G')=0$, where $W$ has been defined in \S\ref{Notation}. We get then the inclusion
\[L(G)\subset L(G').\]
By using the hypothesis $\ell(W+D-G')=0$, one can prove that there exists a vector space $V$ such that $L(G)\subset V\subset L(G')$ and ${\ev_D}_{|V}: V\rightarrow \mathbb{F}_q^n$ is an isomorphism (see \cite{E92}). Hence, we get the following diagram:
\begin{figure}[h]
  \centering
  \begin{tikzpicture}
    \node (A){$L(G)\ \ $};
    \node (C)[node distance=2cm, right of=A]{$\ V\ $};
    \node (E)[node distance=2cm, right of=C]{$\ L(G')$};
    \node (F)[node distance=1.3cm, below of=A]{$C_L(D, G)$};
    \node (G)[node distance=1.3cm, below of=C]{$\ \mathbb{F}_q^n$};
    \draw[->,font=\scriptsize] (A) to node [right]{$\ev_D$} (F);
    \draw[->, font=\scriptsize] (C) to node[right]{$\ev_D$} node [above, rotate=90]{$\backsim$} (G);
    \draw[right hook->, font=\scriptsize] (A) to node [right]{ } (C);
    \draw[right hook->, font=\scriptsize] (C) to node [right]{ } (E);
    \draw[right hook->, font=\scriptsize] (F) to node [right]{ } (G);
  \end{tikzpicture}
  \label{fig:existing}
\end{figure}
\medskip
$ $\\
We can now see the received vector $\yv$ and the error vector $\error$ as vectors of the evaluation of two functions $f_{\error}, f_{\yv}\in L(G')$. We denote by $D_{\error}$ the divisor such that $0\le D_{\error}\le D$ and $P_i\in\supp(D_{\error})$ if and only if $i\in\supp(\error)$ (i.e. $e_i\ne 0$). The aim of Ehrhard's and many other decoding algorithms for algebraic geometry codes is to introduce an additional divisor $F$ with $t+2g\le\deg(F)<n$ and try to compute the space
\begin{equation}\label{loc_space}
L(F-D_{\error}).
\end{equation} 
Indeed the space $L(F-D_{\error})$ is composed by all functions in $L(F)$ locating in some way the error positions. In particular, if $\supp(F)\cap \supp(D_{\error})=\emptyset$, then $L(F-D_{\error})$ is composed by all functions in $L(F)$ which vanish at $\{P_i\mid i\in \supp(\error)\}$. However, we will see soon that this last hypothesis on the support of $F$ is not necessary to the algorithm and that, by adding a simple assumption on the degree of $F$, the very knowledge of an arbitrary nonzero $f\in L(F-D_{\error})$ makes possible to recover $f_{\error}$. First let us consider $\Lambda\in L(F-D_{\error})$. We have
\[\Lambda f_{\yv}=\Lambda f_{\cv}+\Lambda f_{\error},\]
where $\Lambda f_{\cv}\in L(F+G)$ and $\Lambda f_{\error}\in L(F+G'-D)$. One can note that, if it is possible to isolate the second part, that is $\Lambda f_{\error}$, then by dividing by $\Lambda$ we can recover $f_{\error}$. To do so, we want to add an assumption in order to have uniqueness of the decomposition of $\Lambda f_{\yv}$. Hence, we now introduce the following map
\[\delta_{\yv}:\map{L(F)}{L(F+G')}
  {\Lambda}{\Lambda f_{\yv}.}\]
Let us analyse the set $L(F+G')$. One can note that $L(F+G)$, $L(F+G'-D)\subseteq L(F+G')$. Moreover, since $G'\ge G$ and $\supp(G')\cap\mathcal{P}=\emptyset$, we have $L(F+G)\cap L(F+G'-D)=L(F+G-D)$. The assumption we need is then the following one.
\begin{assumption}\label{NB} We assume $\deg(G+F)<n$.
\end{assumption}
\medskip
\noindent
Now, thanks to Assumption \ref{NB}, we get $L(F+G-D)=\{0\}$ and, for a certain vector space $Z_1\subset L(F+G')$,
\begin{equation}\label{dec1}
L(F+G')=L(F+G)\oplus L(F+G'-D)\oplus Z_1.
\end{equation}
\begin{thm}\label{thm_princ}
The very knowledge of an arbitrary $\Lambda\in L(F-D_{\error})\setminus \{0\}$ makes possible to recover $f_{\error}$.
\end{thm}
\begin{proof}
Let us denote by $\pi$ the projection $L(F+G')\rightarrow L(F+G'-D)$ with respect to the decomposition in \eqref{dec1}. As said before, for any $\Lambda\in L(F-D_{\error})$ we have $\Lambda f_{\cv}\in L(F+G)$ and $\Lambda f_{\error}\in L(F+G'-D)$, \\
that is
\begin{equation}\label{propL(F-De)}
\Lambda f_{\yv}\in L(F+G)\oplus L(F+G'-D).
\end{equation}
In particular, since the decomposition is unique, the equality $\Lambda f_{\error}=\pi(\Lambda f_{\yv})$ holds. Therefore, if $\Lambda\ne 0$, we can easily recover $f_{\error}=\frac{\pi(\Lambda f_{\yv})}{\Lambda}\cdot$
\end{proof}
\begin{rem}
One can note that Theorem \ref{thm_princ} holds whenever $\Lambda$ belongs to a space different from $\{0\}$ of the form 
\[L(F-D_{\error'}),\]
where $\supp(\error')\supseteq\supp(\error)$. However, in order to have $L(F-D_{\error'})\ne\{0\}$, we need the support of $\error'$ to be not too large. Indeed, let us consider the designed distance $d^*=n-\deg G$ of the code and suppose $\w(\error')\ge d^*$. By Assumption \ref{NB},  we have 
\[\deg (F - D_{\error'})=\deg F - \w(\error')\le \deg F-n+\deg G<0.\]
Hence, by Proposition \ref{prop:preRR}, $L(F-D_{\error'})=\{0\}$. 
\end{rem}

\subsection{The algorithm}
The problem now is to find a way to compute $L(F-D_{\error})$, without knowing the support of the error vector. In Ehrhard's paper for decoding up to half the designed distance \cite{E93}, the idea is to compute the space
\begin{equation}\label{S(F)_iniz}
S(F)\eqdef \{f\in L(F)\mid \delta_{\yv}(f)\in L(F+G)\oplus L(F+G'-D)\},
\end{equation}
which fulfills the inclusion $L(F-D_{\error})\subseteq S(F)$ (see \eqref{propL(F-De)}), and adapt the divisor $F$ to have the equality. The main result which makes that possible is the following:
\begin{prop}\label{prop:chiave}
Assume $L(F-D_{\error})\ne \{0\}$ and $\deg(F)\le d^{\ast}-g-1$. Then one and only one of the following statements holds:
\begin{itemize}
\item $S(F)=L(F-D_{\error})$;
\item There exists a rational point $P\in\supp(D)$ with $\dim S(F-P)\le \dim S(F)-2$.
\end{itemize}
\end{prop}
\medskip
\noindent
For the proof, see \cite{E93}. We emphasise that there are no hypotheses on the support of $F$ and that in particular the result remains true if $\supp(F)\cap\mathcal{P}\ne \emptyset$. Proposition \ref{prop:chiave} tells us that either we already have $S(F)=L(F-D_{\error})$, or we can construct a new divisor $F-P$ for some $P\in\supp(D)$, such that $\dim S(F-P)$ decreases quite fast with respect to $\dim S(F)$. Furthermore, we have $L(F-P-D_{\error})\subseteq L(F-D_{\error})$, where in particular
\begin{equation}\label{dis:ell}
\ell(F-D_{\error})-1\le \ell(F-P-D_{\error})\le \ell(F-D_{\error}).
\end{equation}
Let us denote by $\{P_{i_m}\}_{m\ge 1}$ the sequence of found points in the support of $D$ such that for any $m\ge 0$ 
\[\dim S(F-\sum_{j=1}^m P_{i_j}-P_{i_{m+1}})\le \dim S(F-\sum_{j=1}^m P_{i_j})-2,\]
and by $F_{m+1}$ the divisor $F_m-P_{i_{m+1}}$, where $F_0\eqdef F$. Hence, since the sequence 
\[\dim S(F)\ge\dim S(F_1) \ge\dim S(F_2) \ge\cdots\] decreases faster than the sequence 
\[\ell(F-D_{\error})\ge \ell(F_1-D_{\error})\ge \ell(F_2-D_{\error})\ge\cdots\]
and $L(F_m-D_{\error})\subseteq S(F_m)$ for any $m$, there will be an equality for some $m$. However, we need to have enough elements $F_m$ in the sequence, in order to have the equality $L(F_m-D_{\error})=S(F_m)$ for one of them. Therefore, the two hypotheses of Proposition \ref{prop:chiave} $L(F_m-D_{\error})\ne \{0\}$ and $\deg(F_m)\le d^*-g-1$ need to be fulfilled for several $F_m$'s, in order to build a long enough sequence. The result in Theorem \ref{thm:halfmd} will emphasise the role of the hypothesis $t\le \frac{d^*-1}{2}$ in this problem. In order to prove this theorem, we will need the following proposition and corollary.
\begin{prop}\label{prop:exact_seq}
Let $\pi$ be the projection $L(F+G')\rightarrow L(F+G'-D)$ with respect to the decomposition in \eqref{dec1}. There is an exact sequence of vector spaces
\begin{figure}[h]
  \centering
  \begin{tikzpicture}
    \node (A){$0$};
    \node (B)[node distance=2cm, right of=A]{$L(F-D_{\error})$};
    \node (C)[node distance=2.7cm, right of=B]{$S(F)$};
    \node (D)[node distance=3cm, right of=C]{$L(G+F-D+D_{\error})$};
    \draw[->,font=\scriptsize] (A) to node [above]{$ $} (B);
    \draw[->,font=\scriptsize] (B) to node [above]{$i$} (C);
    \draw[->,font=\scriptsize] (C) to node [above]{$\Phi$} (D);
  \end{tikzpicture}
  \label{fig:existing2}
\end{figure} 
\medskip
\\
where for any $\Gamma\in S(F)$, $\Phi(\Gamma)=\Gamma f_{\error}-\pi(\Gamma f_{\yv})$.
\end{prop}
\begin{proof}
See \cite{E93}.
\end{proof}
\begin{cor}\label{cor:dim}
We have $\ell(F-D_{\error})\le\dim S(F)\le \ell(F-D_{\error})+\ell(G+F-D+D_{\error})$. 
\end{cor}
\begin{thm}\label{thm:halfmd}
Let $\X$ be a curve of genus $g$ and $C=\AGcode{\X}{\mathcal{P}}{G}$ an algebraic geometry code on $\X$ with designed distance $d^*\ge 6g$. Let $F$ be any divisor of degree $\deg F=t+2g$. Then Algorithm 1 corrects every vector $\yv=\cv+\error$ with $t=\w(\error)\le\frac{d^*-1}{2}\cdot$
\end{thm}
\begin{proof}
The proof can be found in \cite{E93}, but we report it here, since a generalisation of this result will be presented in the next section. Let $F$ be a divisor with $\deg F=t+2g$, where $t=\w(\error)\le \frac{d^*-1}{2}$. We denote by $F_1=F, F_2, F_3\dots$ the sequence of divisors constructed by applying Proposition \ref{prop:chiave}. First, let us prove that this sequence exists, that is, for any $i$ smaller than a certain bound, the hypotheses of Proposition \ref{prop:chiave} hold for $F_i$. One can observe that since by hypothesis $t\le \frac{d^*-1}{2}$ and $d^*\ge 6g$, then 
\begin{equation}\nonumber
d^*-t \ge d^*-\frac{d^*-1}{2}\ge \frac{6g+1}{2}=3g+\frac{1}{2}\cdot
\end{equation}
In particular $t\le d^*-3g-1$. Therefore, for any $m$ we have
\[\deg F_m=2g+t-m\le 2g+t\le d^*-g-1.\]
We now prove that $L(F_m-D_{\error})\ne \{0\}$ for any $m\le g$. By Riemann-Roch theorem we get
\begin{equation}\label{dim_ge_1}
\ell(F_m-D_{\error})\ge t+2g -m-t-g+1\ge 1.
\end{equation}
Therefore the hypotheses of Proposition \ref{prop:chiave} are fulfilled for at least $F_0, \dots, F_g$, which means that we can actually construct this sequence of divisors. We define
\begin{equation}\label{def:delta}
\Delta_m\eqdef \dim S(F_m)-\ell(F_m-D_{\error}).
\end{equation} 
We now show that $\Delta_0\le g$ and that the sequence $\{\Delta_m\}_m$ is strictly decreasing. By Corollary~\ref{cor:dim} we get
\begin{equation}\label{cons_cor}
\Delta_0=\dim S(F)-\ell(F-D_{\error})\le \ell(G+F-D+D_{\error}).
\end{equation}
Since $t\le\frac{d^*-1}{2}$ and $d^*=n-\deg G$, we have in particular
\begin{equation}\label{eq:use_dec_rad}
\deg(G+F-D+D_{\error})=n-d^*+t+2g-n+t=2t-d^*+2g\le 2g-1.
\end{equation}
Now we claim that $\ell(G+F-D+D_{\error})\le g$. If $\deg(G+F-D+D_{\error})=2g-1>2g-2$, we have by Riemann Roch theorem
\[\ell(G+F-D+D_{\error})=\deg(G+F-D+D_{\error})-g+1=2g-1 -g+1=g.\]
Otherwise, if $\deg(G+F-D+D_{\error})\le 2g-2$, by Clifford's Theorem (see Theorem \ref{thm:clifford}), we get
\[\ell(G+F-D+D_{\error})\le 1+\frac{1}{2}\deg(G+F-D+D_{\error})\le g.\]
Finally, we prove that the sequence of the $\Delta_m$'s is stricly decreasing. By using \eqref{dis:ell} and the definition of the $P_{i_m}$'s, we have
\begin{eqnarray*}
\Delta_{m+1} & = & \dim S(F_m-P_{i_{m+1}})-\ell(F_m-P_{i_{m+1}}-D_{\error})\\
			 & \le & \dim S(F_m)-2 - \ell(F_m-D_{\error})+1\\
			 & = & \Delta_m -1.
\end{eqnarray*} 
For any $m\le g$, if $\Delta_m=0$ then the algorithm stops. Otherwise the algorithm constructs $F_{m+1}$ and gets $\Delta_{m+1}\le \Delta_m-1$. In this way, we can contruct for sure at least $g+1$ divisors $F_0, \dots, F_{g}$ and it is enough. Indeed, since $\Delta_0\le g$  and the sequence is strictly decreasing, we will get $\Delta_m=0$ for some $m\le g$. 
\end{proof}
\begin{rem}\label{rem:9293}
In this remark we want to point out the reason why the process of adaptation of the divisor $F$ becomes necessary to reach the decoding radius $\frac{d^*-1}{2}$. In \cite{E92}, the only space $S(F)$ is computed and there is no adaptation process to have $S(F)=L(F-D_{\error})$. Indeed the strategy is rather to find a condition for this equality to hold from the start. This condition turns out to be a bound on the degree of $F$, $\deg F< d^*-t$ (see Proposition 1 at the end of the remark). The price of this bound though, is a limitation on the decoding radius. Indeed, in \cite{E92}, the lower bound for $\deg F$ is $t+g$, which together with the hypothesis in Proposition 1 of \cite{E92}, gives
\[t+g\le \deg F< d^*-t,\] 
that is, we have the decoding radius $t\le\frac{d^*-1-g}{2}$. Once this decoding radius is exceeded, we have $\deg F\ge d^*-t$, the equality between $S(F)$ and $L(F-D_{\error})$ could no longer hold and the process of adaptation of the divisor $F$ becomes necessary to ensure it. Thanks to this process, we have then the improvement of the decoding radius from $\frac{d^*-1-g}{2}$ (\cite{E92}) up to $\frac{d^*-1}{2}$ (\cite{E93}).
\begin{namedthm*}{Proposition 1 \cite{E92}}[Function version]
If $\deg F+t<d^*$, then $L(F-D_{\error})=S(F)$.
\end{namedthm*}
\noindent
 The reader can find the proof in Appendix \ref{sec:appendix}.
\end{rem}
\begin{rem}\label{rem:motiv}
Observe that once $t>\frac{d^*-1}{2}$, we may have $\Delta_0>g$ (see \eqref{eq:use_dec_rad}). In particular, we could need more than $g$ steps to have the equality $S(F_j)=L(F_j-D_{\error})$ for some $j$. Though, if $j>g$, we may have $L(F_j-D_{\error})=\{0\}$ (see \eqref{dim_ge_1}), while for the algorithm to work, we need $L(F_j-D_{\error})\ne\{0\}$ (see Theorem \ref{thm_princ}).
\end{rem}
\begin{algorithm}\label{AlgoBW}
\caption{Ehrhard algorithm - unique decoding}
\textbf{Inputs:} $f_{\yv}=f_{\cv}+f_{\error}\in \mathbb{F}_q^n$ where $\cv\in C$ and $\w(\error)\le\frac{d^*-1}{2}$, $t=\w(\error)$.\\
\textbf{Output:} $f_{\error}\in L(G')$ such that
$\ev_{\mathcal{P}}(f_{\error})=\error$.
\begin{algorithmic}[1]
 \State{Choose $F$ with $\supp(F)\cap\mathcal{P}=\emptyset$ and $\deg F=t+2g$;}
\State{$j\gets 0$ and $F_0\gets F$;}
\State{Look for a point $P\in\{P_1,\dots, P_n\}$ such that $\dim(S(F_j-P))\le\dim(S(F_j))-2$;}
\If{such a point $P$ exists} \State{$F_{j+1}\gets F_j-P$;} 
\State{$j\gets j+1$;} \State{go to Step 3;}
\Else { compute $f_{\error}=\frac{\pi(\Lambda f_{\yv})}{\Lambda}$ for some $\Lambda\in S(F_j)$;}
\EndIf 
  \State \Return $f_{\error}$;
\end{algorithmic}
\end{algorithm}
\begin{rem}
In the process of adaptation of the divisor $F$, the points $P_{m+1}\in\supp(D)$ such that
\[\dim S(F_m-P_{m+1})\le \dim S(F_m)-2,\]
does not belong necessarely to $\supp(D_{\error})$. The tests we made on the generalised algorithm give some evidences of this fact in \S\ref{sec:test}.
\end{rem}
\section{Generalisation of the algorithm to correct more errors ($\ell=2$)}\label{sec:l2}
We want now to solve the decoding problem for $\yv=\cv+\error$ with $\w(\error)>\frac{d^*-1}{2}$. In particular we know that any decoding algorithm for Reed-Solomon codes, decoding beyond half the minimum distance of the code, once generalised to algebraic geometry codes, presents a penalty given by the genus of the curve. For instance, the decoding radius of Sudan algorithm ($\ell=2$) for algebraic geometry codes is 
\begin{equation}\label{dec_rad_Sud}
t_{Sud}=\frac{2n-3\deg G-2}{3}-\frac{2}{3}g.
\end{equation}
(see Appendix \ref{appendixSudan}). We would like to generalise Ehrhard algorithm in order to correct the same amount of errors, without the term in $g$. 
\subsection{Foundation of the algorithm}
The purpose of the algorithm stays the same, that is to introduce a certain divisor $F$ and compute the space $L(F-D_{\error})$. What changes is that we want to be able to find this space even for an amount of errors which is larger than half the designed distance. From Remark \ref{rem:motiv}, we know that in this situation, given $F$ with $\deg F=t+2g$, the gap between $S(F)$ and $L(F-D_{\error})$ could be too large with respect to $g$. This gap decreases by one at each step, but we would like to fill it in $g$ steps. The idea then, is to work with a different space $S(F)$ such that the gap decreases by $\ell$ at each step instead, for a certain parameter $\ell$. In order to do so, we now generalise the foundations of the algorithm. Let us consider $\AGcode{\X}{\mathcal{P}}{G}$ with $G$ as in the previous section and the codes 
\begin{equation}\label{def:AB}
A=\AGcode{\X}{\mathcal{P}}{F}\ \ \ B=\AGcode{\X}{\mathcal{P}}{W+D-G-F}.
\end{equation}
The idea of the generalisation, is based on the following remark.
\begin{rem}\label{rem:S(F)=M}
One can prove that $\ev_{\mathcal{P}}(S(F))\subseteq K_{\yv}$, where $K_{\yv}$ is the set computed in the \textit{error correcting pairs algorithm} (see \cite{P92}), that is
\[K_{\yv}\eqdef\{\av\in A\mid \langle \av\ast\yv, \bv\rangle=0\ \ \forall \bv\in B\}.\]
In \cite{P92}, this set is $\Ker(E_{\wv})$, for a specific linear application $E_{\wv}$ and in \cite[\S $3$]{E92} it is proved that the equality $\ev_D(S(F))=K_{\yv}$ holds whenever $\supp(F)\cap \mathcal{P}=\emptyset$. 
\end{rem}
\noindent
Let us rename the spaces $S(F)$ and $K_{\yv}$ respectively by $S_1(F)$ and $K_{\yv}^{(1)}$. We know that in the generalisation of the error correcting pairs algorithm to correct more errors, that is the \textit{power error locating pairs} (see \cite{CP20}), the intersection of several spaces $K_{\yv}=\cap_{i=1}^\ell K_{\yv}^{(i)}$ is computed rather than the only $K_{\yv}^{(1)}$. Let us consider $\ell=2$ for the moment (we present the general case $\ell\ge 2$ in \S\ref{sec:l>2}). We have%space in terms of projection, like $S_1(F)$ was a reformulation of $K_{\yv}^{(1)}$. We will do that for $\ell=2$ for a  start. Hence, we consider $K_{\yv}^{(1)}$ and $K_{\yv}^{(2)}$ defined as in the PELP
\begin{eqnarray}
K_{\yv}^{(1)}&\eqdef &\{\av\in A\mid \langle \av\ast\yv, \bv\rangle=0\ \ \forall \bv\in B\},\label{M1}\\
K_{\yv}^{(2)}&\eqdef &\{\av\in A\mid \langle \av\ast\yv^2, \bv\rangle=0\ \ \forall \bv\in (B^{\perp}\ast C)^\perp \}. \label{M2}
\end{eqnarray}
As said in Remark \ref{rem:S(F)=M}, we know that, whenever $\supp(F)\cap \mathcal{P}=\emptyset$, $K_{\yv}^{(1)}$ corresponds to $S_1(F)$ by evaluating in the points $P_1, \dots, P_n$ . Let us find a space $S_2(F)$ which reformulates in the same spirit $K_{\yv}^{(2)}$. 
\begin{comment}

First, remember that we denote by $\cv^2$ the product $\cv\ast\cv$. Note that
\[\cv^2\in C_L(2G)\]
hence there exists $f_{\cv^2}\in L(2G)$ such that $\ev_D(f_{\cv^2})=\cv^2$, that is $f_{\cv^2}=f_{\cv}^2$.
\begin{rem}\label{rem2}
Let $G'$ be as in the basic case. Then there exists $\tilde{V}$ such that ${\ev_D}_{|\tilde{V}}:\tilde{V}\rightarrow \mathbb{F}_q^n$ is an isomorphism. Hence we get the graph:
\begin{figure}[h]
  \centering
  \begin{tikzpicture}
    \node (A){$L(2G)\ \ $};
    \node (C)[node distance=2cm, right of=A]{$\ \tilde{V}\ $};
    \node (E)[node distance=2cm, right of=C]{$\ L(2G')$};
    \node (F)[node distance=1.3cm, below of=A]{$C_L(D, 2G)$};
    \node (G)[node distance=1.3cm, below of=C]{$\ \mathbb{F}_q^n$};
    \draw[->,font=\scriptsize] (A) to node [right]{$\ev_D$} (F);
    \draw[->,font=\scriptsize] (C) to node [right]{$\ev_D$} (G);
    \draw[right hook->, font=\scriptsize] (A) to node [right]{ } (C);
    \draw[right hook->, font=\scriptsize] (C) to node [right]{ } (E);
    \draw[right hook->, font=\scriptsize] (F) to node [right]{ } (G);
  \end{tikzpicture}
  \label{fig:existing2}
\end{figure} 
\end{rem}
\noindent
Thanks to Remark \ref{rem2}, we know that there exists $f_{\yv^2}\in L(2G')$ such that $\ev_D(f_{\yv^2})=\yv^2$.

\end{comment}
First, let us consider the vector $\yv^2$. We have seen in \S\ref{sec:Ehrhard} that there exists $f_{\yv}\in L(G')$ such that $\ev_{\mathcal{P}}(f_{\yv})=\yv$. In particular, we get $f_{\yv}^2\in L(2G')$ and $\ev_{\mathcal{P}}(f_{\yv}^2)=\yv^2$. We now denote $f_{\yv}^2$ by $f_{\yv^2}$ and define the following map
\[\delta_{\yv^2}:\map{L(F)}{L(F+2G')}
  {\Lambda}{\Lambda f_{\yv^2}.}\]
One can easily prove that this map is well-defined. We would like to see the space $L(F+2G')$ as the direct sum of some particular subspaces (as for $L(F+G')$ in \S\ref{subs:found1}). As in \S\ref{subs:found1}, we have that both the spaces $L(F+2G)$ and $L(F+2G'-D)$ are included in $L(F+2G')$ and it holds
\[L(F+2G)\cap L(F+2G'-D)=L(F+2G-D).\]
\begin{assumption}\label{NB2}
We assume that $\deg(F+2G)<n$. 
\end{assumption}
\medskip
\noindent
Under Assumption \ref{NB2}, we get $\deg(F+2G-D)<0$, hence $L(F+2G-D)=\{0\}$ and there exists a subspace $Z_2$ of $L(F+2G')$ such that 
\begin{equation}\label{dec2}
L(F+2G')=L(F+2G)\oplus L(F+2G'-D)\oplus Z_2.
\end{equation}
\begin{rem}\label{rem:PowDec}
The idea of decoding at the same time several powers of the same vector $\yv$, follows the one of Sidorenko, Schmidt and Bossert in the so called \textit{power decoding algorithm} for Reed--Solomon codes \cite{SSB10}, inspired in turn by a decoding algorithm of interleaved Reed--Solomon codes. Observe that we are applying here the same procedure. Indeed we are now considering two decoding problems, that is, one with received vector $\yv$ and code $\AGcode{\X}{\mathcal{P}}{G}$ and one with received vector $\yv^2$ and code $\AGcode{\X}{\mathcal{P}}{2G}$. That is why the construction we have just made for $\yv^2$ is equivalent to that for $\yv$ (\S\ref{subs:found1}) but with the divisor $2G'$, the code $\AGcode{\X}{\mathcal{P}}{2G}$ and the received vector $\yv^2$, instead of respectively $G'$, $\AGcode{\X}{\mathcal{P}}{G}$ and $\yv$. From the point of view of applying the algorithm to two received vectors, Assumption \ref{NB2} comes as a natural request to correct the received vector $\yv^2$ as it plays the role of Assumption \ref{NB} with $\yv$. Furthermore this assumption makes easier to compute the decoding radius of our algorithm (see Lemma \ref{lem:dimZ} and Theorem \ref{thm:dec_rad}). Though, since the two decoding problems are related, $\yv^2$ being the square of $\yv$, Assumption \ref{NB2} is not as important as it seems. Indeed, given $\Lambda\in L(F-D_{\error})$, we do not really need the space $L(F+2G')$ to split as in \eqref{dec2} to recover $f_{\error}$, since we already know how to do that by Assumption \ref{NB} together with Theorem \ref{thm_princ}. We want to point out then that Assumption \ref{NB2} is not a necessary condition for the algorithm to work, as shown by our tests in \S\ref{sec:test}.
\end{rem}
\begin{rem}
Note that, since $\deg(G)> 0$, Assumption \ref{NB2} implies Assumption \ref{NB}. 
\end{rem}
\medskip
\noindent
It is actually possible to compute the dimension of the spaces $Z_1$ and $Z_2$.
\begin{lem}\label{lem:dimZ}
Given $Z_1$ as in \eqref{dec1} and $Z_2$ as in \eqref{dec2}, then
\begin{align*}
\dim Z_1&=\deg(D- F - G) +g-1\\
\dim Z_2&=\deg(D- F - 2G) +g-1
\end{align*}
\end{lem}
\begin{proof}
First, we show that $\Omega(F+iG)=\Omega(F+iG'-D)=\{0\}$ for $i=1,2$. Since we took $G'$ such that $\ell(W+D-G')=0$, by Riemann-Roch theorem we have  
\[0=\ell(W+D-G')\ge 2g+n-\deg G'-g+1,\]
that is $\deg G'\ge n+g-1$. Thus, since $\deg F\ge t+g$, we get 
\begin{equation}\label{forOmega0}
\deg(F+2G'-D)>\deg(F+G'-D)>2g-2
\end{equation}
and in particular $\Omega(F+G'-D)=\Omega(F+2G'-D)=\{0\}$. Furthermore, \[\deg(F+2G)>\deg(F+G)\ge t+g+g-1>2g-2,\] 
hence $\Omega(F+2G)=\Omega(F+G)=\{0\}$. Therefore we can compute from \eqref{dec1}
\begin{align*}
\dim Z_1&=\ell(F+G')-\ell(F+G)-\ell(F+G'-D)\\
		&=\deg F+\deg G'-g+1-\deg F - \deg G+g-1 - \deg F - \deg G' +n +g-1\\
		&= \deg(D-F-G)+g-1.
\end{align*}
In the same way from \eqref{dec2} we get $\dim Z_2=\deg(D-F-2G)+g-1$.
\end{proof}

\subsection{The algorithm}
We can now define the space
\begin{equation}\label{def:S2}
S_2(F)\eqdef\{f\in L(F)\mid \delta_{\yv^2}(f)\in L(F+2G)\oplus L(F+2G'-D)\}.
\end{equation}
As for $S_1(F)$, we have that $\ev_{\mathcal{P}}(S(F))\subseteq K_{\yv}^{(2)}$ and under further conditions on the support of the divisor $F$ and $\deg G$, we have the equality.
\begin{thm}\label{thm:S_2=M_2}
Given $S_2(F)$ as in \eqref{def:S2} and $K_{\yv}^{(2)}$ as in \eqref{M2}, if $\supp(F)\cap\mathcal{P}=\emptyset$ and $\deg G\ge 2g$, we have 
\[\ev_{\mathcal{P}}(S_2(F))=K_{\yv}^{(2)}.\]
\end{thm}
\begin{proof}
See Appendix \ref{sec:appendix}.
\end{proof}

\begin{prop}\label{prop:loc_in_S_2}
Let $D_{\error}$ be as in the previous section. Then $L(F-D_{\error})\subset S_2(F)$.
\end{prop}
\begin{proof}
Let us consider $f_{\yv^2}$. We recall that we defined this function to be equal to $f_{\yv}^2$, hence it belongs to $L(2G')$ and fulfills $\ev_{\mathcal{P}}(f_{\yv^2})=\yv^2$. In particular 
\[f_{\yv}^2=f_{\cv}^2+2f_{\cv}f_{\error}+f_{\error}^2,\]
where $f_{\cv}\in L(G)$ and $f_{\error}\in L(G')$. If $\Lambda\in L(F-D_{\error})$, then we get 
\begin{eqnarray*}
(\Lambda f_{\cv}^2)&\ge &-F-2G,\\
(\Lambda f_{\error}^2)&\ge & -F+D-2G',\\
(\Lambda f_{\cv}f_{\error})&\ge & -F+D-G-G'\ge -F+D-2G',
\end{eqnarray*}
since $G'>G$ and in particular $f_{\error}\in L(G'-D+D_{\error})$.
\end{proof}
\medskip
\noindent
Let us finally introduce the space
\begin{equation}\label{S(F)_new}
S(F)\eqdef S_1(F)\cap S_2(F).
\end{equation}
Thanks to Proposition \ref{prop:loc_in_S_2}, we have $L(F-D_{\error})\subseteq S(F)$. As in Ehrhard's paper, the idea is now to close the gap between these two spaces. Hence, the next task is to adapt Proposition \ref{prop:chiave} to the $S(F)$ we have just constructed. To do so, we have to adapt two lemmas. The proofs of these two lemmas and of the adaptation of Proposition \ref{prop:chiave} come directly from the proofs of Lemma 1, Lemma 2 and Proposition 8 in \cite{E93}, though we write them here anyway for sake of completeness. 
\begin{lem}\label{lem1}
If $L(F-D_{\error})\ne S(F)$, then there exist at most $\deg(F+D_{\error})-d^{\ast}$ rational points $P\in~\supp(D_{\error})$ such that $S(F)\subseteq L(F-P)$.
\end{lem}
\begin{proof}
Without loss of generality, after reindexing, one can suppose $P_1,\dots, P_m$ to be the points in $\supp(D_{\error})$ such that 
\[S(F)\subseteq L(F-P_i)\ \ \ \forall i=1,\dots, m.\]
In particular we have $S(F)\subseteq \bigcap_{i=1}^m L(F-P_i)=L(F-\tilde{D})$, where $\tilde{D}=\sum_{i=1}^m P_i$. Let us consider 
\[\Gamma\in S(F)\setminus L(F-D_{\error})\subseteq S_1(F)\setminus L(F-D_{\error}).\] 
By Proposition \ref{prop:exact_seq}, we have $\Phi(\Gamma)=\Gamma f_{\error} - \pi(\Gamma f_{\yv})\ne 0$ where $\pi(\Gamma f_{\yv})\in L(F+G'-D)$. We get then
\[(\Gamma f_{\error}-\pi(\Gamma f_{\yv}))\ge\min(-F+\tilde{D}-G'+D-D_{\error}, -F-G'+D)=-G'-F+D-(D_{\error}-\tilde{D}).\]
By definition of $\Phi$ in Proposition \ref{prop:exact_seq} we get in particular 
\[0\ne \Gamma f_{\error}-\pi(\Gamma f_{\yv})\in L(G'+F-D+(D_{\error}-\tilde{D}))\cap L(G+F-D+D_{\error})=L(G+F-D+(D_{\error}-\tilde{D})).\]
Hence, $\deg(G+F-D+(D_{\error}-\tilde{D}))\ge 0$, that is 
\[m\le \deg(F+D_{\error})-d^{\ast}.\]
\end{proof}
\begin{lem}\label{lem2}
If $L(F-D_{\error})\ne \{0\}$, then there are at most $g$ rational points $P\in\supp(D_{\error})$ such that $S_1(F-P)\cap S_2(F)=S(F)\cap L(F-P)$.
\end{lem}
\begin{proof}
As a consequence of Riemann Roch theorem, there are at most $g$ points $P$ in $\supp(D_{\error})$ such that $L(F-D_{\error}-P)=L(F-D_{\error})$.  Indeed again without loss of generality we can suppose that $P_1, \dots, P_m$ are the points in $\supp(D_{\error})$ such that $L(F-D_{\error})=L(F-D_{\error}-P_i)$ for all $i=1, \dots, m$. In particular we have
\[L(F-D_{\error})=L\left( F-D_{\error}-\sum_{i=1}^mP_i\right).\]
Furthermore, by Proposition \ref{prop:ell},
\[\ell(F-D_{\error})\le \ell(F-D_{\error})-\ell\left(\sum_{i=1}^m P_i\right)+1\le\ell(F-D_{\error})-m+g-1+1.\]
Hence, $m\le g$. Now, it suffices to prove that given a point $P\in\supp(D_{\error})$
\[L(F-D_{\error}-P)\ne L(F-D_{\error})\implies S_1(F-P)\cap S_2(F)\ne S(F)\cap L(F-P).\]
Let us consider $\Gamma\in L(F-D_{\error})\setminus L(F-D_{\error}-P)$. In particular we have $\Gamma\in S(F)\cap L(F-P)$. We now show that $\Gamma\notin S_1(F-P)$. If it were, we would have $\Gamma f_{\yv}=g+h$ with $g\in L(F+G-P)$ and $h\in L(F+G'-D-P)$. On the other hand, $\Gamma f_{\yv}=\Gamma f_{\cv}+\Gamma f_{\error}\in L(F+G)\oplus \L(F+G'-D)$. Both decompositions are in $L(F+G)\oplus L(F+G'-D)$, hence the uniqueness gives
\begin{equation}\label{contr}
h=\Gamma f_{\error}\in L(F+G'-D-P).
\end{equation}
Though note that $v_P(\Gamma f_{\error})= v_P(\Gamma)+v_P(f_{\error})= -v_P(F)$, against \eqref{contr}.
\end{proof}
\begin{prop}\label{prop:chiave2}
Assume $L(F-D_{\error})\ne \{0\}$ and $\deg F\le d^{\ast}-g-1$. Let $S(F)$ be as in \eqref{S(F)_new}. Then one and only one of the following statement holds:
\begin{itemize}
\item $S(F)=L(F-D_{\error})$
\item There exists a rational point $P\in\supp(D)$ with $\dim(S(F-P))\le \dim (S(F))-2$.
\end{itemize}
\end{prop}
\begin{proof}
If $S(F)=L(F-D_{\error})$, then for any point $P\in\supp(D)$, we get
\[\dim S(F-P)\ge \ell(F-P-D_{\error})\ge \ell(F-D_{\error})-1=\dim S(F)-1.\]
Now, if $S(F)\ne L(F-D_{\error})$, by Proposition \ref{prop:exact_seq} we have $L(G+F-D+D_{\error})\ne \{0\}$. Hence we get $\deg(G+F-D+D_{\error})\ge 0$ that is $\deg(F+D_{\error})- d^{\ast}\ge 0$. Hence by applying Lemma \ref{lem1} and Lemma \ref{lem2} and using the hypothesis $\deg(F)\le d^{\ast}-g-1$, we get that there exist at least
\[\deg D_{\error}-g-\deg(F+D_{\error})+d^{\ast}=d^{\ast}-\deg F-g\ge 1\]
points $P$ in $\supp(D_{\error})$ such that $S(F)\nsubseteq L(F-P)$ and $S_1(F-P)\cap S_2(F)\ne S(F)\cap L(F-P)$. Hence for such a point we get
\[S(F-P)\subseteq S_1(F-P)\cap S_2(F)\nsubseteq S(F)\cap L(F-P)\nsubseteq S(F). \]
In particular $\dim S(F-P) \le \dim S(F)-2$.
\end{proof}
\medskip
\noindent
Thanks to Proposition \ref{prop:chiave2}, the sequence $\{\Delta_i\}_{i\ge 0}$ defined in \eqref{def:delta} verifies $\Delta_{i+1}\le \Delta_{i}-1$. As said in the beginning of the section, we would like this sequence to decrease faster and it is clear that the faster decreases the sequence $\{\dim S(F_i)\}_{i\ge 0}$, the faster decreases the sequence $\{\Delta_{i}\}_{i\ge 0}$.
\begin{rem}
Let us consider $F$ a generic divisor in $\{F_i\}_{i\ge 0}$, $P\in\supp(D_{\error})$ and $\Lambda\in S_1(F-P)\cap S_2(F)$. We want to understand whether $\Lambda\in S_2(F-P)$.  We get by definition of $S_1(F-P)$ and $S_2(F)$:
\begin{align}
&\Lambda f_{\yv}\in L(F-P+G)\oplus L(F-P+G'-D)\nonumber\\
&\Lambda f_{\yv^2}\in L(F-2G)\oplus L(F+2G'-D)\label{appart}.
\end{align}
In particular, $\Lambda f_{\yv}=g+h$ with $g\in L(F-P+G)$ and $h\in L(F-P+G'-D)$. Let us analyse $\Lambda f_{\yv^2}=(g+h)(f_{\cv}+f_{\error})$:
\begin{eqnarray*}
(g f_{\cv})&\ge & -F+P-2G,\\
(g f_{\error})&\ge & -F+P-G-G'+D-D_{\error}\ge -F+P-2G'+D-D_{\error},\\
(h f_{\cv})&\ge & -F+P-G'+D-G\ge -F+P-2G'+D,\\
(h f_{\error})&\ge & -F+P-G'+D-G'+D-D_{\error}\ge -F+P-2G'+D.
\end{eqnarray*}
In particular, we have $g f_{\cv}, h f_{\cv}, h f_{\error}\in L(F-P+2G)\oplus L(F-P+2G'-D)$, while 
\[g f_{\error}\in L(F-P-2G'-D+D_{\error}).\]
Therefore, by \eqref{appart}, given $\Lambda\in S_1(F-P)\cap S_2(F)$ we have 
\[\Lambda\in S_2(F-P)\iff g f_{\error}\in L(F-P+2G)\oplus L(F-P+2G'-D).\]
In particular, if $\Lambda$ verifies this property for any $\Lambda\in S_1(F-P)\cap S_2(F)$, then 
\[S_1(F-P)\cap S_2(F-P)=S_1(F-P)\cap S_2(F).\]
\end{rem}
\begin{comment}
\begin{algorithm}\label{Algo}
\caption{Power decoding algorithm with $\ell=2$}
\textbf{Inputs:} $\yv\in\mathbb{F}_q^n$ with $\yv=\cv+\error$,
$t=\w(\error)$, $G$ divisor
\textbf{Output:} some $f\in L(G')$ such that
$\w(\ev_{D}(f))\le t$ and $f_{\yv}-f\in L(G)$ or failure
\begin{algorithmic}[1]
\State{Choose $F$ with $\supp(F)\cap\supp(D)=\emptyset$;}
\State{$j\gets 1$ and $F_1\gets F$;}
\State{Look for a point $P\in\{P_1,\dots, P_n\}$ such that $\dim(S(F_j-P))\le\dim(S(F_j))-2$;}
\If{such a point $P$ exists} \State{$F_{j+1}\gets F_j-P$;} 
\State{$j\gets j+1$;} \State{go to Step 3;}
\Else { compute $f=\frac{\pi(\Lambda f_{\yv})}{\Lambda}$ for some $\Lambda\in S(F_j)$;}
\EndIf 
\State {\Return $f$;}
\end{algorithmic}
\end{algorithm}
\end{comment}
\paragraph*{Empirical Behavior:} we observed that for a random error vector, we have
\[\dim S(F_m-P_{m+1})\le \dim S(F_m)-3.\]
In particular, it seems the three strict inclusions to be the following
\[S(F_m-P_{m+1})=S_1(F_m-P_{m+1})\cap S_2(F_m-P_{m+1})\subsetneq S_1(F_m-P_{m+1})\cap S_2(F_m)\subsetneq S(F_m)\cap L(F_m-P_{m+1})\subsetneq S(F_m).\]
The second and third inclusions correspond to the two inclusions of respectively Lemma \ref{lem1} and Lemma \ref{lem2}, while the first one seems to depend strictly on the chosen error vector. For random vectors, it is easy to find points such that the three inclusions are strict, while for a ``worst case'', that is, when we have two codewords at the same distance from $\yv$, we could not find such a point and we got 
\[\dim S(F_m-P_{m+1})=\dim S(F_m)-2.\]

\medskip
\noindent
Now that we have all the ingredients, we can describe the algorithm (actually the only thing that changes with respect to Algorithm 1 is that we use the new notion of $S(F)$) and try to compute its decoding radius by adapting the proof of Theorem 1 of \cite{E93}. Since we want to correct more than half the designed distance of the code and the algorithm gives back only one solution, we do not look for a sufficient condition for the algorithm to work, but rather for a necessary one.

\begin{comment}

\begin{rem}
In this remark we are going to estimate the situations where the first inclusion is not strict, \ip{piuttosto proviamo che tutte le situazioni in cui gli elementi non nulli dell'errore sono i valori dell'elemento del codice, sono delle situazioni in cui la dimensione cala di due} that is 
\[S_1(F-P)\cap S_2(F-P)=S_1(F-P)\cap S_2(F).\]
Let us consider $f\in S_1(F-P)\cap S_2(F)$.

Since $ff_{\yv^2}\in S_2(F-P)$, we have that $f_{\error}\alpha\in L(F-P+2G)\oplus L(F-P+2G'-D)$. A situation where this happens is when $e_i=-f_{\cv}(P_i)$ for any $i\in \supp(\error)$. Indeed in this case $f_{\yv}(P_i)=0$ for any $P_i\in \supp(D_{\error})$ and $f_{\yv}\in L(F-D_{\error})$
\end{rem}

\end{comment}
\begin{thm}\label{thm:dec_rad}
Assume $\deg F=t+2g$ with $t\le d^{\ast}-3g-1$. If the error vector is such that $S(F)$ verifies the empirical behavior and $\deg(2G+F)<n$, then a necessary condition for the algorithm to work is 
\begin{equation}\label{dec_rad}
t\le\frac{2n-3\deg G-2}{3}\cdot
\end{equation}
\end{thm}
\begin{proof}
Let us start by observing that, since $t\le d^{\ast}-3g-1$, we have
\begin{equation}\nonumber
\deg F_j=\deg F-j\le \deg F=t+2g\le d^{\ast}-g-1.
\end{equation}
 Note that as for Theorem \ref{thm:halfmd}, for any $j\le g$, it holds
\[\ell(F_j-D_{\error})\ge t+2g-j-t-g+1\ge 1.\]
Hence, if necessary\footnote{If we do not find $S(F_j)=L(F_j-D_{\error})$ for some $j<g$} we can apply Proposition \ref{prop:chiave} to $F_j$ for any $1\le j\le g$ and construct a sequence of divisors of length at least $g+1$. Let us consider again the quantity
\[\Delta_j=\dim S(F_j)-\ell(F_j-D_{\error}),\]
where $\{F_j\}$ is the sequence of constructed divisor, that is, $F_{j+1}=F_j-P_{i_{j+1}}$. We claim that, if $\Delta_0\le 2g$, then $\Delta_j=0$ for some $j\le g$. Indeed as said before, we have $\ell(F_{j+1}-D_{\error})\ge \ell(F_j-D_{\error})-1$ and by hypothesis $\dim S(F_{j+1})\le \dim S(F_j)-3$, hence
\[\dim S(F_{j+1})-\ell(F_{j+1}-D_{\error})\le \dim S(F_j)-\ell(F_j-D_{\error})-2.\]
Therefore the sequence of the $\Delta_j$ is strictly decreasing and $\Delta_j=0$ for some $j\le g$. Now we want to find a necessary condition to have $\Delta_0\le 2g$, that is to have
\begin{equation}\label{cond_nec}
\ell(F-D_{\error})+2g\ge\dim S(F).
\end{equation}
In order to do so, we want to bound $\dim S(F)$. It is possible to write $S(F)$ in the following way
\begin{equation}\label{writingS(F)}
S(F)=\{f\in L(F)\mid \pi_{Z_1}(\delta_{\yv}(f))=\zerov\ \land\ \pi_{Z_2}(\delta_{\yv^2}(f))=\zerov\},
\end{equation}
where $\pi_{Z_1}$ and $\pi_{Z_2}$ are respectively the projections $L(F+G')\rightarrow Z_1$ and $L(F+2G')\rightarrow Z_2$ with respect to the decompositions \eqref{dec1} and \eqref{dec2}. In particular, $S(F)$ is composed by the elements of $L(F)$ which fulfill certains conditions, therefore we can bound
\begin{equation}\label{dimS(F)}
\dim S(F)= \ell(F)-\#\text{conditions}\ge \ell(F) - \dim Z_1 - \dim Z_2.
\end{equation}
Therefore, putting together the condition on $\Delta_0$ \eqref{cond_nec} and \eqref{dimS(F)}, we get
\begin{equation}\label{quasifin}
\ell(F-D_{\error})+2g\ge\dim S(F)\ge \ell(F)-\dim Z_1-\dim Z_2.
\end{equation}
Now, by Lemma \ref{lem:dimZ}, we have
\begin{equation}\nonumber
\dim Z_1=n-t-g -\deg G-1\ \ \ \dim Z_2=n-t-g-2\deg G-1
\end{equation}
By substituting the values of $Z_1$ and $Z_2$ in \eqref{quasifin} and applying Riemann-Roch theorem we get
\[t\le\frac{2n-3\deg G-2}{3}\cdot\]
\end{proof}

\section{Generalisation to $\ell\ge 2$}\label{sec:l>2}
In this section we will show how to generalise the strategy we have seen in \S\ref{sec:l2} to build an algorithm with parameter $\ell\ge 2$. Experimentally, the decoding radius of this algorithm reaches the amount
\[\frac{2n\ell-\ell(\ell+1)\deg G-2\ell}{2(\ell+1)}\]
which is the decoding radius of Sudan algorithm without any penalty in the genus of the curve (see Appendix \ref{appendixSudan}).
\subsection{Foundation of the algorithm}
As in the cases $\ell=1, 2$, a divisor $F$ with certain properties is introduced and the aim of the algorithm is to find the space $L(F-D_{\error})$. For that, once we fix $\ell$, we need to define a space $S(F)$ wich contains $L(F-D_{\error})$ and such that, given a specific sequence of divisors $\{F_j\}_j$, the gap 
\[\Delta_j=\dim S(F_j)-\ell(F_j-D_{\error})\]
decreases fast enough with respect to $g$. Again, the following assumption comes naturally if we think we are applying the basic algorithm to the first $\ell$ powers of $\yv$ and will help to estimate the decoding radius of the algorithm. We recall though that it is not a  necessary condition for the algorithm to work (see $(\ast)$ test for $\ell=3$ in \S\ref{sec:test}).
\begin{assumption}\label{NB>2} We assume that $\deg(F+\ell G)<n$;
\end{assumption}
\medskip
\noindent
Observe that, since $\deg G>0$, by Assumption \ref{NB>2} we have $\deg(F+iG)<n$ for any $i=1, \dots, \ell$. Therefore, for any $i=1, \dots, \ell$ there exists $Z_i\subseteq L(F+iG)$ such that the following equalities hold
\begin{eqnarray*}
L(F+G')&=& L(F+G)\oplus L(F+G'-D)\oplus Z_1\\
L(F+2G')&=& L(F+2G)\oplus L(F+2G'-D)\oplus Z_2\\
\dots &=& \dots\\
L(F+\ell G')&=& L(F+\ell G)\oplus L(F+\ell G'-D)\oplus Z_{\ell}.
\end{eqnarray*}
\begin{lem}\label{lem:dimZ>2}
Given $Z_i\subset L(F+iG')$ such that $L(F+iG')=L(F+iG)\oplus L(F+iG'-D)\oplus Z_i$, then 
\[\dim Z_i=\deg (D-F-G)+g-1.\] 
\end{lem}
\begin{proof}
The proof is an easy generalisation of the one for $\ell=2$ (see Lemma \ref{lem:dimZ}).
\end{proof}
\medskip
\noindent
For any $i=1, \dots, \ell$, we define $f_{\yv^i}\eqdef f_{\yv}^i\in L(iG')$ and the map 
\[\delta_{\yv^i}:\map{L(F)}{L(F+iG')}
  {\Lambda}{\Lambda f_{\yv^i}.}\]

\subsection{The algorithm}
It is possible now to define for any $i=1, \dots, \ell$ the space
\begin{equation}\label{def:Si}
S_i(F)\eqdef\{f\in L(F)\mid \delta_{\yv^2}(f)\in L(F+iG)\oplus L(F+iG'-D)\}.
\end{equation}
\begin{rem}
Again, we have $\ev_{\mathcal{P}}(S_i(F))\subseteq K_{\yv}^{(i)}$, where, given $A, B$ as in \eqref{def:AB},
\begin{equation}\label{def:Kyi}
K_{\yv}^{(i)} = \{\av\in A \mid \langle\av\ast\yv, \bv\rangle=0\ \forall \bv\in (B^{\perp}\ast C^{i-1})^{\perp}\}.
\end{equation}
Furthermore, it is possible to generalise Theorem \ref{thm:S_2=M_2} to every $i\le\ell$.
\begin{thm}
Given $S_i(F)$ as in \eqref{def:Si} and $K_{\yv}^{(i)}$ as in \eqref{def:Kyi}, if $\supp(F)\cap \mathcal{P}=\emptyset$ and $\deg G\ge 2g$, we have
\[\ev_{\mathcal{P}}(S_i(F))=K_{\yv}^{(i)}.\]
\end{thm}
\begin{proof}
The proof is exactly the same as for $\ell=2$. See Appendix \ref{sec:appendix}.
\end{proof}
\end{rem}
\begin{prop}\label{LFDeInSi}
For any $i\le\ell$, we have $L(F-D_{\error})\subseteq S_i(F)$.
\end{prop}
\begin{proof}
The proof is an easy adaptation of the proof of Proposition \ref{prop:loc_in_S_2}. Let $\Lambda\in L(F-D_{\error})$ and $i\le \ell$. We know that $\delta_{\yv^i}(\Lambda)=\Lambda f_{\yv}^i=\sum_{j=0}^i \binom{i}{j} \Lambda f_{\cv}^jf_{\error}^{i-j}$. Now, we treat the term with $j=i$ and separately the others with $j<i$:
\begin{itemize}
\item[$j=i:$] $\ (\Lambda f_{\cv}^i)\ge  -F-iG$
\item[$j<i:$] $\ (\Lambda f_{\cv}^j f_{\error}^{i-j})\ge  -F+D_{\error}-jG-(i-j)G'+(i-j)(D-D_{\error})\ge -F-iG'+D$
\end{itemize}
that is $\delta_{\yv^i}(\Lambda)\in L(F+iG)\oplus L(F+iG'-D)$ and $\Lambda\in S_i(F)$.
\end{proof}
\medskip
\noindent
It is now possible to define the space $S(F)$ for this choice of $\ell$:
\begin{equation}\label{def:Snew}
S(F)\eqdef \bigcap_{i=1}^\ell S_i(F).
\end{equation}
Thanks to Proposition \ref{LFDeInSi}, we have $L(F-D_{\error})\subseteq S(F)$. We want now to close the gap between the two spaces. One can observe that Lemma \ref{lem1}, Lemma \ref{lem2} and Proposition \ref{prop:chiave2} can be generalised straightforwardly to the $S(F)$ defined in \eqref{def:Snew}. In particular Lemma \ref{lem2} changes in the following way.
\begin{lem}\label{lem2>2}
If $L(F-D_{\error})\ne \{0\}$, then there are at most $g$ rational points $P\in\supp(D_{\error})$ such that 
\[S_1(F-P)\cap \bigcap_{i=2}^{\ell} S_i(F)=S(F)\cap L(F-P).\]
\end{lem}
\medskip
\noindent
By Proposition \ref{prop:chiave2}, it is possible then to build a sequence $F_0=F, F_1, F_2, \dots$ such that for any~$j\ge 0$ 
\begin{equation}\label{fallS}
\dim S(F_{j+1})\le \dim S(F_j)-2.
\end{equation}
As for the case $\ell=2$, among the following inclusions,
\[S(F_j-P)=\bigcap_{i=1}^{\ell}S_i(F_j-P)\subseteq S_1(F_j-P)\cap \bigcap_{i=1}^{\ell}S_i(F_j)\subsetneq S(F)\cap L(F_j-P)\subsetneq S(F),\]
the last two are the ones that give the gap in \eqref{fallS} and entail then $\Delta_{j+1}\le \Delta_j-1$. In order for the sequence $\{\Delta_j\}_j$ to decrease faster, we need more strict inclusions between $S(F_j-P)$ and~$S(F_j)$.
\paragraph*{Empirical Behavior:} as in the case $\ell=2$, we observed that the dimension of $S(F_j)$ decreases faster than expected when a random error vector is considered. In particular we got
\begin{equation}\label{empiricalb}
\dim S(F_j-P)\le \dim S(F_j)-(\ell+1)
\end{equation}
which implies $\Delta_{j+1}\le \Delta_j-\ell$. The further $\ell-1$ strict inclusions which cause this drop in dimension are the following
\[S(F_j-P)\subsetneq \bigcap_{i=1}^{\ell-1}S_i(F_j-P)\cap S_{\ell}(F_j)\subsetneq \bigcap_{i=1}^{\ell-2}S_i(F_j-P)\cap S_{\ell-1}(F_j)\cap S_{\ell}(F_j) \subsetneq\dots\subsetneq S_1(F_j-P)\cap \bigcap_{i=2}^{\ell}S_i(F_j).\]
\medskip
\noindent
It is now possible to compute the decoding radius of the algorithm for $\ell\ge 2$. To do so, we generalise Theorem~\ref{thm:dec_rad}.
\begin{thm}
Assume $\deg F=t+2g$ with $t\le d^{\ast}-3g-1$. If the error vector is such that $S(F)$ verifies the empirical behavior and $\deg(F+\ell G)<n$, then a necessary condition for the algorithm to work is 
\begin{equation}\label{dec_rad}
t\le\frac{2\ell n-\ell(\ell+1)\deg G-2\ell}{2(\ell+1)}\cdot
\end{equation}
\end{thm}
\begin{proof}
The proof is almost the same as in the case with $\ell=2$. The only difference consists in the necessary condition to fill the gap between $S(F)$ and $L(F-D_{\error})$ in $g$ steps. Indeed we impose here the condition 
\[\Delta_0\le \ell g\]
instead of $\Delta_0\le 2g$. An estimate for the dimension of $S(F)$ can be deduced by \eqref{def:Si} and is
\begin{equation*}
\dim S(F)\ge \ell(F) - \sum_{i=1}^\ell\dim Z_i,
\end{equation*}
where $\dim Z_i$ is given by Lemma \ref{lem:dimZ>2}.
\end{proof}
\section{Some experimentations}\label{sec:test}
In this section we first propose some guidelines on the parameters of the algorithm for $\ell=2,3$ and then we give some experimental observations from the tests we made.
\subsection{The parameters}
In order to test the algorithm with the right parameters, we need the genus $g$ of the curve, the number of evaluation points $n$ and the degree of the divisor $G$ to fulfill several conditions:
\begin{itemize}
\item[(\textit{i})] $t\le\frac{2\ell n-\ell(\ell+1)\deg G-2\ell}{2(\ell+1)}$ (decoding radius \eqref{dec_rad})
\item[(\textit{ii})] $\deg(F+\ell G)<n$ (Assumption \ref{NB2})
\item[(\textit{iii})] $\deg F\le d^*-g-1$ (hypothesis in Theorem \ref{thm:dec_rad})
\item[(\textit{iv})] $\frac{2(\ell-1) n-\ell(\ell-1)\deg G-2(\ell-1)}{2\ell}< \frac{2\ell n-\ell(\ell+1)\deg G-2\ell}{2(\ell+1)}$ (decoding radius$(\ell-1)<$ decoding radius$(\ell)$)
\item[(\textit{v})] $\deg G\ge g-1$
\end{itemize}
First, notice that Theorem \ref{thm:dec_rad} holds for all $F$ with $t+2g\le\deg F\le d^*-g-1$. Here, we will just study the parameters for $\deg F=t+2g$. Moreover, we want to be able to run the algorithm up to its decoding radius, hence we set 
\[t=\frac{2\ell n-\ell(\ell+1)\deg G-2\ell}{2(\ell+1)}\cdot\] 
By imposing these conditions on $\deg F$ and $t$, and developing (\textit{iv}), (\textit{i-iv}) become:
\begin{itemize}
\item[(\textit{i})] $t=\frac{2\ell n-\ell(\ell+1)\deg G-2\ell}{2(\ell+1)}$
\item[(\textit{ii})] $2n-\ell(\ell+1)\deg G-4g(\ell+1)\ge 0$
\item[(\textit{iii})] $2n+(\ell-2)(\ell+1)\deg G-6g\ell-6g-2\ge 0$
\item[(\textit{iv})] $2n-\ell(\ell+1)\deg G-2(\ell^2+\ell+1)\ge 0$.
\end{itemize}
In particular, notice that (\textit{ii}) implies (\textit{iii}) and (\textit{iv}) when $g>\frac{l}{2}$ and $\deg G\ge 0$. That means that for $\ell=2$ we can run the algorithm on codes which fulfill:
\begin{equation*}
g-1\le \deg G\le \frac{n-6g}{3},
\end{equation*}
while for $\ell=3$ we need our code to satisfy
\begin{equation*}
g-1\le \deg G\le\frac{n-8g}{6}\cdot
\end{equation*}

\subsection{Some tests}
In Table \ref{fig:test} and Table \ref{fig:fail} there are listed some results about the algorithm's behavior with $\ell=2,3$. We worked with the following three curves:
\begin{enumerate}
\item \label{1}$X^6 + Y^6+XZ^5=0$ on $\mathbb{F}_{7^3}$;
\item \label{2}$X^8 -YZ^7-ZY^7=0$ on $\mathbb{F}_{7^2}$;
\item \label{3}$ZY^5-X^6-XZ^5-Z^6=0$ on $\mathbb{F}_{11^3}$;
\end{enumerate}
Looking at Table \ref{fig:test} and Table \ref{fig:fail}, if $C=\AGcode{\X}{\mathcal{P}}{G}$ is the code we are running the algorithm on, we indicate by $q$ the cardinality of the field, $\X$ the curve, $g$ the genus of $\X$, $n$ the length of the code, that is $n=|\mathcal{P}|$, $\deg G$ the degree of the divisor $G$.  Moreover we list the values of half the designed distance of the code (column $\frac{d^{\ast}-1}{2}$) and of the decoding radius of Sudan algorithm (column ``Sudan''), in order to compare them with the decoding radius of the new algorithm (``dec.radius''). For each test, we pick randomly an error vector $\error$ with $\w(\error)=t$ for a specific $t$ and run the algorithm on $\yv=\cv+\error$ with a random $\cv\in C$ and with a power parameter $\ell$. The value of $t$ will be underlined when the decoding radius of the new algorithm is exceeded. We denote by ``pts'' the set of  points $\{P_{i_j}\}\subseteq\supp(D)$, such that
\[\dim S(F_{j-1}-P_{i_j})\le \dim S(F_{j-1})-2. \]
In particular, for any test, we check if the points which guarantee this gap in the dimension belong to the support of $D_{\error}$. We recall that $\Delta_0=\dim S(F) - \ell(F-D_{\error})$, where $F$ is the initial divisor with $\deg F=t+2g$ and that, if the algorithm verifies the empirical behavior \eqref{empiricalb}, then $\Delta_{i+1}-\Delta_i\le \ell$. Finally in Table \ref{fig:test} we list the tests where the algorithm succeeds, while in Table \ref{fig:fail} there are some cases where the algorithm fails. %Finally in the last column on the table we will find a ``Y'' if the algorithm succeded and a ``N'' otherwise.
\begin{table}[H]
\begin{center}
\begin{tabular}{|c|c|c|c|c|c|c|c|c|c|c|c|c|c|}
\hline
$\ell$ &$q$ & $\X$ & $g$ & $n$ & $\deg G$ & $\frac{d^*-1}{2}$ &Sudan & dec. radius & $t$ & pts$\subseteq D_{\error}$ & $\Delta_0$ & $\Delta_{i+1}-\Delta_i$ \\
\hline
$2$ & $7^3$ & \ref{1} & $10$ & $200$ & $2g-1$ & $90$ & $107$ & $113$ & $113$& true & $18$ & $2$ \\
\hline
$2$ & $7^3$ & \ref{1} & $10$ & $200$ & $ \frac{n-6g-1}{3}$ & $76$ & $80$ & $86$ & $86$ & true & $18$ &  $2$\\
\hline
$2$ & $7^2$ & \ref{2} & $21$ & $230$ & $2g-1$ & $94$ & $98$ & $111$ & $111$ & false & $40$ & $2$ \\
\hline 
$3$ & $7^3$ & \ref{1} & $10$ & $200$ & $2g-2$ & $90$ & $113$ & $120$ & $120$ & true & $27$ & $3$ \\
\hline
$3$ & $7^3$ & \ref{1} & $10$ & $200$ & $\frac{n-8g-1}{6}-1$ & $90$ & $115$ & $122$ & $122$ & false & $29$ & $3$ \\
\hline
$2$ & $7^3$ & \ref{1} & $10$ & $200$ & $50$ ($\ast$) & $72$ & $76$ & $82$ & $82$ & false & $18$ & $2$ \\
\hline
$3$ & $7^3$ & \ref{1} & $10$ & $200$ & $\frac{n-1}{6}-3$ ($\ast$)& $84$ & $97$ & $104$ & $104$ & true & $24$ & $2,3\times 4$ \\
\hline
$2$ & $11^3$ & \ref{3} & $10$ & $200$ & $\frac{n-6g-1}{3}-10$ & $81$ & $90$ & $96$ & $96$ & false & $18$ & $2$ \\
\hline
\end{tabular}
\caption{Tests on the algorithm for $\ell=2$, $\ell=3$. \label{fig:test}}
\end{center}
\end{table}
\paragraph*{Comments:} First we want to point out that, whenever the parameter are chosen to satisfy (\textit{i-iv}), then $\Delta_0\le \ell g$. That was not free, as we recall the decoding radius bound was a necessary condition to have $\Delta_0\le \ell g$ and not a sufficient one. Furthermore, we can see that it is actually possible to correct up to the decoding radius, which is then larger than Sudan decoding radius. In particular the gap $\Delta_0$ reduces by expected at every step by $\ell$ for both $\ell=2,3$, that is the algorithm satisfies the hypothesis of empirical behavior \eqref{empiricalb} and we get to have $\Delta_j=0$ and $L(F_j-D_{\error})\ne\{0\}$ for some $j\le g$. One can observe that pts is not always contained in the support of $D_{\error}$. Moreover it is really difficult to find a point which does not fulfill
\begin{equation}\label{gap}
\dim S(F_i-P)\le \dim S(F_i)-2,
\end{equation}
and that once a point $P$ which does fulfill \eqref{gap} is found for the first step, it will satisfy it also for the next steps, that is, in our sequence $\{F_j\}_j$ we have $F_j=F-jP$ for every $j\le g$. Finally, observe that the cases with the symbol $(\ast)$ are the only cases where not all the bounds (\textit{i-iv}) hold. In particular we have $\deg(F+\ell G)\ge n$, that is, it is no longer sure that the spaces 
\[L(F+\ell G), \ \ \ \ \ \ L(F+\ell G'-D)\]
are in direct sum. In this situation it is possible to use the following modified notion of $S_{\ell}(F)$ 
\[S_{\ell}(F)=\{f\in L(F)\mid \delta_{\yv^{\ell}}(f)\in L(F+\ell G)+L(F+\ell G'-D)\}.\]
In the $(\ast)$ case with $\ell=2$, we observed that actually the two spaces are still in direct sum and the gap $\Delta_i$ decreases by $\ell$. Hence the algorithm works up to the decoding radius. In the $(\ast)$ case with $\ell=3$, the algorithm works anyway even if $\Delta_i$ does not decrease by $\ell=3$ at every step, but most of the times by $2$. That is, almost at every step we have $\dim S(F-P)\le \dim S(F)-3$, where the inclusion which is not strict is the first from the left in the following sequence
\[S(F-P)\subseteq S_1(F-P)\cap S_2(F-P)\cap S_3(F)\subsetneq S_1(F-P)\cap S_2(F)\cap S_3(F)\subsetneq S(F)\cap L(F-P)\subsetneq S(F).\]
\subsection{Failure cases}
\begin{table}[H]
\begin{center}
\begin{tabular}{|c|c|c|c|c|c|c|c|c|c|c|c|c|c|}
\hline
$\ell$ &$q$ & $\X$ & $g$ & $n$ & $\deg G$ & $\frac{d^*-1}{2}$ &Sudan & dec. radius & $t$ & pts$\subseteq D_{\error}$ & $\Delta_0$ & $\Delta_{i+1}-\Delta_i$ \\
\hline
$2$ & $7^3$ & \ref{1} & $10$ & $200$ & $2g-1$ & $90$ & $107$ & $113$ & $\underline{114}$ & true & $21$ & $2$\\
\hline
$2$ & $7^2$ & \ref{2} & $21$ & $230$ & $2g-1$ & $94$ & $98$ & $111$ & $\underline{112}$ & false & $43$ & $2$ \\
\hline
$2$ & $11^3$ & \ref{3} & $10$ & $200$ & $\frac{n-6g-1}{3}$ & $76$ & $80$ & $86$ & $86$ & false & $14$ & $1$ \\
\hline
$3$ & $7^3$ & \ref{1} & $10$ & $200$ & $25$ & $67$ & $ 104 $ & $111$ & $96$ & false & $25$ & $1$ \\
\hline
\end{tabular}
\caption{Failure cases. \label{fig:fail}}
\end{center}
\end{table}
\paragraph*{Comments:} We report here four cases where the algorithm does not work. Actually one should not consider all of them as failure cases, as the amount of error exceeds the decoding radius of the algorithm in the first two of them. One can see that in these situations, $\Delta_0>\ell g$. Hence, although the gaps $\Delta_{i+1}-\Delta_i$ are the good ones, by the time $\Delta_i=0$ we have $L(F_i-D_{\error})=\{0\}$ as well and the algorithm fails as expected. In the two last cases all parameters are bounded as requested for the algorithm to work, but unlike the other tests, here the choice of the error vector is not random. Indeed it has been chosen in order to have two solutions $\cv_1, \cv_2\in C$ such that 
\[\dd(\yv, \cv_1)=\dd(\yv, \cv_2)=t.\]
In these cases, the empirical behavior \eqref{empiricalb} is not fulfilled, indeed $\Delta_i$ decreases only by $1$ and no point in $\mathcal{P}$ can make $\Delta_i$ decrease faster. In particular here we only have two strict inclusions given by Proposition \ref{prop:chiave2} and the following chain of equalities
\[S(F_j-P)= \bigcap_{i=1}^{\ell-1}S_i(F_j-P)\cap S_{\ell}(F_j)= \bigcap_{i=1}^{\ell-2}S_i(F_j-P)\cap S_{\ell-1}(F_j)\cap S_{\ell}(F_j) =\dots= S_1(F_j-1)\cap \bigcap_{i=2}^{\ell}S_i(F_j).\]
Hence, even if $\Delta_0\le\ell g$, $g$ steps are not enough to find $S(F_i)=L(F_i-D_{\error})$.

\nocite{*}
\bibliographystyle{alpha}
\bibliography{ms}

\appendix
\section{On the decoding radius of Sudan algorithm}\label{appendixSudan}
This section mainly comes from the ideas of Peter Beelen and shows how to get an improved decoding radius for Sudan algorithm with respect to \cite[\S$2.6$]{BH08}. This improvement mainly consists in analysing the parameters of the linear system to get Sudan polynomial $Q$. Given a curve of genus $g$, we consider a code $C=\AGcode{\X}{\mathcal{P}}{G}$, where $G$ is a divisor with $2g-2<\deg G <n$ and $\mathcal{P}=\{P_1, \dots, P_n\}$. Let us suppose a vector $\yv=\cv+\error$ is given, where $\w(\error)=t$ and there exists $f\in L(G)$ such that
\begin{equation}\label{f}
\cv=(f(P_1),\dots, f(P_n)).
\end{equation}
We denote by $I$ the support of the error vector $I=\supp(\error)$ (in particular $|I|=t$). Let $F$ be a divisor with $\deg F=n-t-1$.
\paragraph{Original problem (Sudan):} given $\ell\ge 1$, find a polynomial $Q(\xv, y)=Q_0(\xv)+Q_1(\xv)y+\dots+Q_{\ell}(\xv)y^{\ell}$ such that
\begin{itemize}
\item[(\textit{i})] $Q_i(\xv)\in L(F-iG)$ for all $i=0, \dots, \ell$
\item[(\textit{ii})] $Q(P_j, y_j)=0$ for all $j=1,\dots, n$.
\end{itemize}
This problem can be solved with a linear system of $n$ equations in $\sum_{i=0}^{\ell}\ell(F-iG)$ unknowns. Hence the system has nonzero solutions if 
\begin{equation}\label{rayon_base_m=1}
t\le\frac{2n\ell-\ell(\ell+1)\deg(G)-2}{2(\ell+1)}-g.
\end{equation}
Now we want to show that this decoding radius can be actually optimised. To do so, we consider the following problem.
\paragraph{Modified problem (Sudan):} Given $f$ as in \eqref{f}, find a polynomial 
\begin{equation}
Q(\xv, y)=(y-f(\xv))(\tilde{Q}_0(\xv)+\tilde{Q}_1(\xv)y+\dots+\tilde{Q}_{\ell-1}(\xv)y^{\ell-1}),
\end{equation}
such that, if we denote by $\tilde{Q}(\xv, y)$ the factor $\tilde{Q}_0(\xv)+\tilde{Q}_1(\xv)y+\dots+\tilde{Q}_{\ell-1}(\xv)y^{\ell-1}$,
\begin{itemize}
\item[(\textit{i'})] $\tilde{Q}_i(\xv)\in L(F-(i+1)G)$ for all $i=0, \dots, \ell-1$
\item[(\textit{i''})] $\tilde{Q}(P_j, y_j)=0$ for all $j\in I$.
\end{itemize}
It is clear that if the modified problem has a solution, then the original problem has one too. This problem can be solved, as the previous one, by a linear system. This time, we have a system of $t$ equations in $\sum_{i=1}^{\ell}\ell(F-iG)$ unknowns. Therefore it admits nonzero solutions if 
\begin{equation}
t\le \frac{2n\ell-\ell(\ell+1)\deg(G)-2}{2(\ell+1)}-\frac{\ell g}{\ell+1}\cdot
\end{equation}

\section{Some technical results}\label{sec:appendix}
Most of the proofs presented in this appendix, are straightforward adaptations of proofs of \cite{E92} and \cite{E93} to $S_2(F)$ or to the language of functions rather than differentials, but we decided to report them here for sake of completeness.
%\begin{namedthm*}{Proposition 1 \cite{E92}}[Function version]
%If $\deg F+t<d^*$, then $L(F-D_{\error})=S(F)$.
%\end{namedthm*}
\begin{prop}{\cite[Proposition 1]{E92}}(Function version)
If $\deg F+t<d^*$, then $L(F-D_{\error})=S(F)$.
\end{prop}
\begin{proof}
We recall that we consider here the $S(F)$ defined in \eqref{S(F)_iniz}. We already know that $L(F-D_{\error})\subseteq S(F)$. Hence we consider now $\Lambda\in S(F)$ and we want to show that $\Lambda\in L(F-D_{\error})$. In order to do so, we first prove that $\Lambda f_{\error}\in L(F+G'-D)$. Since $\Lambda\in S(F)$, there exist $g\in L(F+G)$ and $h\in L(F+G'-D)$ such that $\Lambda f_{\yv}=g+h$. Furthermore $f_{\yv}=f_{\cv}+f_{\error}$, hence
\[\Lambda f_{\cv}-g=h -\Lambda f_{\error}.\]
By way of contradiction let us suppose that $h \ne\Lambda f_{\error}$. Since $f_{\error}\in L(G'-D+D_{\error})$, we have
\begin{align*}
(h-\Lambda f_{\error})\ge & \min(-F-G'+D, -F-G'+D-D_{\error})=-F-G'+D-D_{\error}\\
(\Lambda f_{\cv}- g)\ge & \min(-F-G, -F-G)=-F-G.
\end{align*}
Hence, in particular 
\[(h-\Lambda f_{\error})\ge \max(-F-G'+D-D_{\error}, -F-G)=-F-G+D-D_{\error},\]
that is $h -\Lambda f_{\error}\in L(F+G-D+D_{\error})$. Though, by hypothesis we have
\[\deg (F+G-D+D_{\error})=\deg F+ \deg G-n+t< 0,\]
that is $L(F+G-D+D_{\error})=\{0\}$, which is a contradiction since we supposed $h\ne\Lambda f_{\error}$. Now we know that $\Lambda f_{\error}\in L(F+G'-D)$, we can conclude the proof. First, since $\ev_{\mathcal{P}}(f_{\error})=\error$, we observe that if $P\in\supp(D_{\error})$, then
\[f_{\error}\in L(G'-D+D_{\error})\setminus L(G'-D+D_{\error}-P).\]
In particular, for any $P\in\supp(D_{\error})$, since $\Lambda f_{\error}\in L(F+G'-D)$ and $v_P(f_{\error})=0$, we get
\[v_{P}(\Lambda)=v_{P}(\Lambda)+v_{P}(f_{\error})=v_{P}(\Lambda f_{\error})\ge -v_{P}(F)+1,\]
that is $\Lambda\in L(F-D_{\error})$. 
\end{proof}
\begin{thm}
Given $S_2(F)$ as in \eqref{def:S2} and $K_{\yv}^{(2)}$ as in \eqref{M2}, if $\supp(F)\cap\mathcal{P}=\emptyset$ and $\deg G\ge 2g$, we have 
\[\ev_{\mathcal{P}}(S_2(F))=K_{\yv}^{(2)}.\]
\end{thm}
\medskip
\noindent
In order to prove this theorem we need the following result.
\begin{prop}\label{prop:for_nothing_new}
Let us consider $G'$ as in \S\ref{subs:found1}, that is $G'\ge G$ and $L(W+D-G')=\{0\}$ and let $\varphi$ be the map $\varphi : L(F+2G')\longrightarrow \Omega(F+2G-D)^{\vee }$ where, for any $f\in L(F+2G')$, given $\omega\in \Omega(F+2G-D)$, \[\varphi(f)(\omega)=\sum_{i=1}^n \Res_{P_i}(f\omega).\]
If $\supp(F)\cap\mathcal{P}=\emptyset$, then $\varphi_{|Z_2}: Z_2\rightarrow \Omega(F+2G-D)^{\vee }$ is an isomorphism.
\end{prop}
\begin{proof}
This proof is an adaptation of the proof of Remark \ref{rem:S(F)=M} given in \cite{E93}. It is composed by the following steps:
\begin{itemize}
\item[(1)] $\varphi$ is surjective;
\item[(2)] $L(F+2G'-D)\oplus L(F+2G)\subseteq\Ker(\varphi)$;
\item[(3)] $\dim(Z_2)=\dim L(W+D-2G-F)$.
\end{itemize}
In order to prove that $\varphi$ is surjective, we first show that it suffices to prove the surjectivity of the map $\tilde{\varphi}: L(F+2G')\rightarrow (\mathbb{F}_q^n)^{\vee}$, where for any $h\in L(F+2G')$, given $\av\in \mathbb{F}_q^n$,
\[\tilde{\varphi}(f)(\av)=\sum_{i=1}^n a_i\ev_{P_i}(f).\]
Let us suppose then that $\tilde{\varphi}$ is surjective. One can easily see that the map
\[\Psi:\map{\Omega(F+2G-D)}{\mathbb{F}_q^n}
  {\omega}{(\Res_{P_1}(\omega), \dots, \Res_{P_n}(\omega)).}\]
is injective, its kernel being the space $\Omega(F+2G)$ which is equal to $\{0\}$ as 
\[\deg(F+2G)>2g-2.\] 
Hence, $\Psi$ being injective, its transpose $\Psi^T:(\mathbb{F}_q^n)^{\vee}\rightarrow \Omega(F+2G-D)^{\vee}$ is surjective. By the hypothesis on the surjectivity of $\tilde{\varphi}$, the composition of $\tilde{\varphi}$ and $\Psi^T$, gives a surjective map.
\begin{figure}[h]
\centering
\begin{tikzpicture}
    \node (A){$L(F+2G')$};
    \node (B)[node distance=3cm, right of=A]{$(\mathbb{F}_q^n)^{\vee}$};
    \node (C)[node distance=3.5cm, right of=B]{$\Omega(F+2G-D)^{\vee}.$};
    \draw[->>,font=\scriptsize] (A) to node [above]{$\tilde{\varphi}$} (B);
    \draw[->>,font=\scriptsize] (B) to node [above]{$\Psi^T$} (C);
  \end{tikzpicture}
\end{figure}
We claim that this map is equal to $\varphi$: for any $f\in L(F+2G)$ and $\omega\in \Omega(F+2G-D)$ the following equalities hold
\begin{eqnarray*}
\Psi^T(\tilde{\varphi}(f))(\omega)&=&(\tilde{\varphi}(f)\circ\Psi)(\omega)\\
								  &=&\tilde{\varphi}(f)(\Res_{\mathcal{P}}(\omega))\\
								  &=&\sum_{i=1}^n \Res_{P_i}(f\omega)=\varphi(f)(w).
\end{eqnarray*}
Hence, we now prove that $\tilde{\varphi}$ is surjective. Let us consider $(\error_i^{\vee})_i$ the canonical basis of $(\mathbb{F}_q^n)^{\vee}$. First we claim that for any $i=1, \dots, n$, the set 
\[L(2G'+F-D+P_i)\setminus L(2G'+F-D)\ne\emptyset.\]
To see that, notice that as we proved in \eqref{forOmega0}, we have $\Omega(F+2G'-D+P_i)=\Omega(F+2G'-D)=\{0\}$, and
\begin{eqnarray}
\ell(2G'+F-D+P_i)&=&\deg(2G'-D+F+P_i)-g+1\\
\ell(2G'+F-D)&=&\deg(2G'+F-D)-g+1,
\end{eqnarray}
hence $L(F+2G'-D+P_i)\setminus L(F+2G'-D)\ne\emptyset$. Now it suffices to note that for any $h$ in this set, there is $\lambda\in\mathbb{F}_q^{\ast}$ such that $\tilde{\varphi}(h)=\lambda\error_i^{\vee}$, hence (1) is proved. Let us now consider $h\in L(F+2G'-D)$. Given $\omega\in \Omega(F+2G-D)$ we have $v_{P_i}(h\omega)\ge 0$ for any $i=1, \dots, n$, hence $\Res_{\mathcal{P}}(h\omega)=0$ for any $\omega\in \Omega(F+2G-D)$, that is $\varphi(h)=0$. We consider now $h\in L(F+2G)$. For any $\omega\in \Omega(F+2G-D)$,
\[(h\omega)\ge -D,\]
thus we get $\varphi(h)(\omega)=\sum_{i=1}^n\Res_{P_i}(h\omega)=\sum_{P\in\X}\Res_P(h\omega)=0$ from the residue Theorem. Hence we proved (2). We now finally prove (3). We have
\begin{eqnarray*}
\dim Z_2&=& g-1-\deg(F+2G-D)\\
		&=& \dim\Omega(F+2G-D),
\end{eqnarray*} 
where in the first equality we used Lemma \ref{lem:dimZ}, while in the second one, we use Assumption \ref{NB2}.
\end{proof}
\begin{rem}\label{remannex}
Observe that by \eqref{dualAG} and Proposition \ref{prop:star_prod_AGcodes}, if $\deg G\ge 2g$,
\[\Res_{\mathcal{P}}(\Omega(F+2G-D))=C_{\Omega}(F+2G)=C_L(W+D-F-2G)=(B^{\perp}\ast C)^{\perp},\] 
where $B$ is defined in \eqref{def:AB}.
\end{rem}
\medskip
\noindent
We can now prove Theorem \ref{thm:S_2=M_2}.
\begin{proof}
Let $\pi_{Z_2}$ be the projection $L(F+2G')\rightarrow Z_2$ with respect to the decomposition of the space $L(F+2G')=L(F+2G)\oplus L(G+2G'-D)\oplus Z_2$. We then have
\[S_2=\{\Gamma\in L(F)\mid \pi_{Z_2}\circ\delta_{\yv^2}(\Gamma)=0\}.\]
In particular, by Proposition \ref{prop:for_nothing_new}, for any $\Gamma\in L(F)$ we have 
\[\Gamma\in S_2(F)\iff\pi_{Z_2}\circ\delta_{\yv^2}(\Gamma)=0\iff\varphi_{|Z_{2}}\circ\pi_{Z_2}\circ\delta_{\yv^2}(\Gamma)=0,\]
that is if and only if, for any $\omega\in\Omega(F+2G-D)$ 
\begin{equation}\label{eq}
(\varphi_{|Z_{2}}\circ\pi_{Z_2}\circ\delta_{\yv^2}(\Gamma))(\omega)=0
\end{equation}
Note that, by Proposition \ref{prop:for_nothing_new}, $\varphi(L(F+2G)\oplus L(F+2G'-D))=0$, therefore for any $f\in L(F+2G')$ the following equality holds
\[\varphi(f)=(\varphi_{|Z_2}\circ\pi_{Z_2})(f).\]
Hence, the left hand side of the equation in \eqref{eq} becomes
\begin{eqnarray*}
(\varphi_{|Z_{2}}\circ\pi_{Z_2}\circ\delta_{\yv^2}(\Gamma))(\omega)&=& \varphi(\delta_{\yv^2}(\Gamma))(\omega)=\sum_{i=1}^n\Res_{P_i}(\omega\Gamma  f_{\yv^2})\\
 		&=&\sum_{i=1}^n \Gamma(P_i)\Res_{P_i}(\omega)y_i^2
\end{eqnarray*}
By Remark \ref{remannex}, we have 
\[\{(\Res_{P_1}(\omega), \dots, \Res_{P_n}(\omega))\mid \omega\in \Omega(F+2G-D)\}= (B^{\perp}\ast C)^{\perp}\] 
where $B$ is defined in $\eqref{def:AB}$. Hence, for any $\Gamma\in L(F)$, $\Gamma\in S_2(F)$ if and only if $\ev_{D}(\Gamma)\in K_{\yv}^{(2)}$.
\end{proof}

\end{document}